\documentclass[preprint,12pt]{elsarticle}

\usepackage{cite}
\usepackage{amsmath,amssymb,amsfonts}
\usepackage{graphicx}
\usepackage{textcomp}
\usepackage{xcolor}
\def\BibTeX{{\rm B\kern-.05em{\sc i\kern-.025em b}\kern-.08em
    T\kern-.1667em\lower.7ex\hbox{E}\kern-.125emX}}
    
\usepackage{bm}
\usepackage[misc]{ifsym}
\usepackage{multirow}
\usepackage{stmaryrd}
\usepackage{amsthm}
\usepackage{subfigure}
\usepackage{algorithm}
\usepackage[noend]{algpseudocode}
\usepackage{enumitem}
\usepackage{booktabs}
\usepackage[inkscapelatex=false]{svg}
\usepackage{hyperref}
\usepackage{soul}
\usepackage{tikz}
\usepackage[misc]{ifsym}
\makeatletter
\algnewcommand{\LineComment}[1]{\Statex \hskip\ALG@thistlm \(\triangleright\) #1}
\makeatother

\newtheorem{theorem}{Theorem}

\newtheorem{lemma}{Lemma}

\newtheorem{definition}{Definition}

\newtheorem{problem}{Problem}

\newcommand{\tabincell}[2]{\begin{tabular}{@{}#1@{}}#2\end{tabular}}

\begin{document}
\begin{frontmatter}



\title{Efficient Top-$k$ $s$-Biplexes Search over Large Bipartite Graphs with Time Complexity Guarantees} 


\author[aff1]{Zhenxiang	Xu \corref{cor1}}
\ead{zhenxiangxu@std.uestc.edu.cn}

\author[aff1,aff2]{Yiping Liu \corref{cor1}}
\ead{yliu823@aucklanduni.ac.nz}

\author[aff1]{Yi Zhou\corref{cor2}}
\ead{zhou.yi@uestc.edu.cn }

\author[aff1]{Yimin Hao}
\ead{yiminhao@std.uestc.edu.cn}

\author[aff3]{Zhengren Wang}
\ead{wzr@stu.pku.edu.cn}

\cortext[cor1]{Equal contribution.}

\cortext[cor2]{Corresponding author.}

\affiliation[aff1]{organization={School of Computer Science and Engineering, University of Electronic Science and Technology of China},
            city={Chengdu},
            country={China}}
 \affiliation[aff2]{organization={School of Computer Science,
 University of Auckland},
             addressline={262 Khyber Pass Road, Newmarket},
             city={Auckland},
             country={New Zealand}}
\affiliation[aff3]{organization={School of Computer Science,
Peking University},
            city={Beijing},
            country={China}}

\begin{abstract}
In a bipartite graph, a subgraph is an $s$-biplex if each vertex of the subgraph is adjacent to all but at most $s$ vertices on the opposite set. 
The enumeration of $s$-biplexes from a given graph is a fundamental problem in bipartite graph analysis.
However, in real-world data engineering, finding all $s$-biplexes is neither necessary nor computationally affordable. 
A more realistic problem is to identify some of the largest $s$-biplexes from the large input graph.
We formulate the problem as the {\em top-$k$ $s$-biplex search (TBS) problem}, which aims to find the top-$k$ maximal $s$-biplexes with the most vertices, where $k$ is an input parameter.  
We prove that the TBS problem is NP-hard for any fixed $k\ge 1$.
Then, we propose a branching algorithm, named MVBP, that breaks the simple $2^n$ enumeration algorithm.
Furthermore, from a practical perspective, we investigate three techniques to improve the performance of MVBP: 2-hop decomposition, single-side bounds, and progressive search. 
Complexity analysis shows that the improved algorithm, named FastMVBP, has a running time $O^*(\gamma_s^{d_2})$, where $\gamma_s<2$, and $d_2$ is a parameter much smaller than the number of vertex in the sparse real-world graphs, e.g. $d_2$ is only $67$ in the AmazonRatings dataset which has more than $3$ million vertices.
Finally, we conducted extensive experiments on eight real-world and synthetic datasets to demonstrate the empirical efficiency of the proposed algorithms.
In particular, FastMVBP outperforms the benchmark algorithms by up to three orders of magnitude in several instances.

\end{abstract}







\end{frontmatter}


\section{Introduction}
Bipartite graphs play crucial roles in modelling interactions between two distinct types of entity, making them essential in a variety of applications such as recommendation \citep{maier2022bipartite}, e-commerce \citep{wang2022efficient}, and social networks \citep{zhao2022finding}.
In bipartite graphs, vertices are partitioned into two disjoint sets. Edges do not connect any pair of vertices within the same set. 

An $s$-biplex is a bipartite graph in which each vertex of a set is not adjacent to at most $s$ vertices in the other set. 
In particular, the $0$-biplex is called a biclique. The $s$-biplex model has applications in various fields.

\smallskip
\noindent {\bf Fraud detection.} Online shopping has become a common lifestyle, with many people selecting products based on user ratings. However, fraudulent users may collaborate to artificially promote or demote products for financial gain. These collaborative behaviours can be detected using $s$-biplexes \citep{allahbakhsh2013collusion,beutel2013copycatch}. 
In addition, biplexes can also be applied to detect fake news \citep{gangireddy2020unsupervised,yu2021graph}.  

\smallskip
\noindent {\bf Community detection.} 
People in the same communities often share common interests and preferences. Identifying communities is beneficial for making accurate recommendations and targeted advertising.  
These communities can be identified via $s$-biplexes \citep{maier2022bipartite,murata2004discovery,maier2021biclique}.

\smallskip
\noindent {\bf Biological data analysis.} Interaction analysing is a significant technique to help understand the functions of biological components, such as genes and proteins.
A usual goal of gene expression analysis is to group genes according to their response under different conditions.
Genes that respond to a number of common conditions are considered a significant gene group. This group of genes can be identified via $s$-biplexes \citep{cheng2000biclustering,yang2005improved}.
Similarly, the protein-protein interactions between HIV-1 and humans can be clustered via $s$-biplexes \citep{bustamam2020application,dey2020graph}.

Hence, querying these $s$-biplexes has been a fundamental task \citep{li2008maximal,sim2009mining,sim2011case,dhawan2019spotting,wang2008recommending,wu2015solving,tan2020top}. 
However, listing all $s$-biplexes from bipartite graphs is rarely asked in real applications. 
On the one hand, there can be an exponential number of $s$-biplexes in real world large graphs. 
It is unaffordable to enumerate all these $s$-biplexes using the limited computation resources \citep{luo2022maximum}.
On the other hand, in real-world bipartite graphs, there are many trivially small $s$-biplexes, which do not carry useful information \citep{yu2023maximum}. 
Due to these factors, the {\em top-$k$ $s$-biplex search (TBS) problem}, which only asks for the $k$ maximal $s$-biplexes with most vertices where $k$ is an input parameter, is more realistic.
Here, an $s$-biplex is maximal if there does not exist a larger $s$-biplex that contains it. 
The maximality condition excludes trivially small $s$-biplexes.

\begin{figure}
    \centering
    \includegraphics[width=0.5\linewidth]{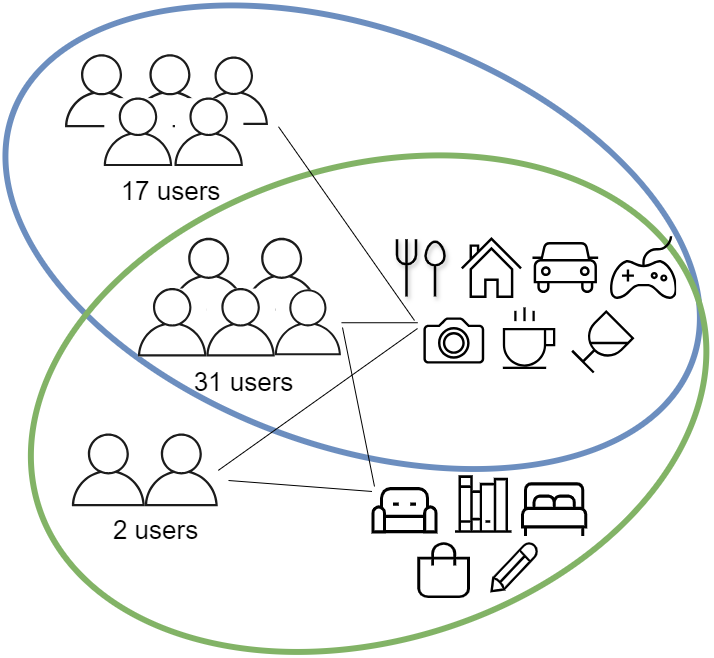}
    \caption{The comparison between the maximal $1$-biplex with most vertices (named TBS, top blue circle) and the maximal $1$-biplex with most edges (named TBSE, bottom green circle) in the Amazon Ratings dataset. 
    This bipartite graph contains product ratings from the Amazon online shopping website.
    We ask the biplexes to contain at least $7$ products to identify users who share common interests with at least $7$ products.
    The two biplexes share 31 common users and $7$ common products.
    The biplex with the most edges identifies dense connections by introducing $5$ additional products and $2$ additional users. In contrast, the biplex with the most vertices identifies $17$ more users who share common interests with the $7$ products. Note that both two biplex models are able to specify the lower-bound number of products. 
    When exploring more users who share common interests in a given number of products, the biplex with the most vertices is preferred.}
    \label{fig:Comparison most edges}
\end{figure}

The {\em TBS} problem differs from another problem investigated by \citep{yu2023maximum}. 
The {\em TBS} problem in this paper aims to find $k$ maximal $s$-biplexes with most vertices.
In contrast, the problem in \citep{yu2023maximum}, namely \emph{top-$k$ $s$-biplex with most edges} (TBSE), asks for the $k$ maximal $s$-biplexes with most edges. 
Theoretically, any (vertex-) induced subgraph of an $s$-biplex remains an $s$-biplex, known as the hereditary property. However, the property of the edge-induced subgraph of an $s$-biplex is not clear. In practice, to explore users that share common interests, the biplex with most vertices is preferred as shown in Fig.\ref{fig:Comparison most edges}.
To the best of our knowledge, we are the first to formulate the {\em TBS} problem.

The {\em TBS} problem is challenging due to the exponential number of candidate solutions. A simple algorithm would require at least $O(2^n)$ time to evaluate each candidate. Indeed, we will prove that the {\em TBS} problem is NP-hard. However, it is possible to solve the {\em TBS} problem by leveraging efficient algorithms designed for two similar scenarios.
The first scenario involves listing $s$-plexes in non-bipartite graphs, where each vertex in an $s$-plex is not adjacent to at most $s$ vertices. The state-of-the-art algorithm for this is the {\em ListPlex} algorithm \citep{wang2022listing}. Extending {\em ListPlex} to search for $s$-biplexes, however, requires additional time to account for vertices on both sets of the bipartite graph, and it does not fully utilise the properties of bipartite graphs.
The second scenario involves searching for maximal $s$-biplexes with the most edges in bipartite graphs. The state-of-the-art algorithm for this is {\em FastBB} \citep{yu2023maximum}. However, the pruning and reduction rules of {\em FastBB} may not apply when searching for top-$k$ maximal $s$-biplexes with the most vertices. Additionally, the running time of {\em FastBB} is dependent on the cube of the maximum degree $\Delta$, which can be quite large in dense graphs, rendering {\em FastBB} impractical in such cases.
We include both two algorithms as our benchmarks.
In this paper, we address the {\em TBS} problem with an algorithm that operates in $O^*(\gamma_s^{d_2})$ time\footnote{The $O^*$ notation hides polynomial factors.}, where $d_2$ denotes the largest 2-hop degeneracy. 
In real-world graphs, $d_2 \leq \Delta^2$ is a small parameter. 


\subsection{Our contributions}

To the best of our knowledge, we are the first to formulate the {\em TBS} problem aiming to search $k$ maximal $s$-biplexes with most vertices. 
Compared to existing models \citep{yu2023maximum}, the TBS model excels at identifying users with shared interests, as illustrated in Fig.\ref{fig:Comparison most edges}.
Theoretically, we prove that the {\em TBS} problem is NP-hard by a reduction from the constrained minimum vertex cover problem \citep{chen2003constrained}.

We propose a non-trivial exact branching algorithm named MVBP.
The MVBP algorithm is based on the Bron-Kerbosch framework \citep{bron1973algorithm}, which is adapted to fit the specific requirements of this problem.
Building on the Bron-Kerbosch framework, we identify key properties of the \emph{TBS} problem and leverage these to develop effective reduction and branching rules. Without these enhancements, a brute-force solution to the \emph{TBS} problem would have a running time of $O^*(2^n)$. By applying advanced analysis techniques for exponential branching algorithms, we show that the inclusion of these reduction and branching rules reduces the algorithm’s running time to $O^*(\gamma_s^n)$ where $n$ is the vertex number of the input bipartite graph and $\gamma_s<2$ is a parameter related to $s$.

We investigate techniques that effectively improve the practical performance of the algorithm. 
Specifically, we introduce graph decomposition (Section~\ref{sec: 2-hop decomposition}), reduction bounds (Section~\ref{sec:reductions with single side bounds}), and progressive search (Section~\ref{sec: progressive search}). {
Compared to existing methods, our work introduces several key improvements:
For graph decomposition: We propose a novel 2-hop decomposition that reduces the size of decomposed subgraphs from the existing $O(\Delta^3)$ complexity \citep{yu2023maximum} to $O(\Delta^2)$. This optimization is particularly impactful when the maximum degree $\Delta$ is large, leading to a significant reduction in the algorithm’s running time.
For effective reductions: While existing algorithms typically use a single bound for both sides of the vertex set \citep{lyu2020maximum,dai2024efficient}, we introduce independent bounds for each side. These single-side bounds are more frequently triggered, allowing for more aggressive branch reduction.
We also incorporate a progressive search strategy that heuristically adjusts the lowerbound and upperbound of target size, further improving the algorithm's efficiency.}
As a result, the running time of the algorithm with these techniques is improved to $O^*(\gamma_s^{d_2})$.
In other words, the exponent part of the running time is reduced from $n$ to $d_2$, a very small parameter in real graphs. For example, $n=3376972$ while $d_2=67$ when finding $1$-biplexes with both sides contains more than $15$ vertices in the AmazonRatings dataset.

We conduct extensive experiments on both real-world datasets and synthetic datasets. (1) The experimental results highlight the performance of the final algorithm, namely FastMVBP. In particular, FastMVBP outperforms all benchmarks by three orders of magnitude in several cases. For example, when searching for $1$-biplexes with at least $3$ vertices on each bipartite vertex set in the YouTube dataset, FastMVBP completes the task in 10.644 seconds, FastBB takes 532.3 seconds, and ListPlex requires 14260.9 seconds.
(2) In a social recommendation task, TBS can identify more users who share common interests than another model.

\section{Preliminaries}
\label{sec_prelimiaries}
In this paper, we focus on an undirected and unweighted bipartite graph $G=(L\cup R,E)$, where $L$ and $R$ are two distinct vertex sets such that no edge connects two vertices within the same set, and $E\subseteq L\times R$ is the set of edges. 
Let $V:=L\cup R$ denote the set of vertices, $n:=|V|$ and $m:=|E|$ denote the number of vertices and the number of edges, respectively.
For clarity, we refer to the set $L$ as left-side vertex set and the set $R$ as right-side vertex set. 



Given $u\in L$ and $v\in R$, we denote the undirected edge between $u$ and $v$ by $(u,v)$.
Given a vertex set $S \subseteq V$, we define $N_S(u) = \{v \in S \mid (u, v) \in E \}$ as the set of neighbours of $u \in V$ within $S$, and {  $\overline{N}_S(u) = \{v \in S \mid (u, v) \in  L\times R\setminus E\}$ as the set of anti-neighbours of $u \in V$ within $S$}. 

When the context is clear, we simply denote $N_V(u)$ as $N(u)$.

The degree of $u$ in $G$ is defined as $d(u):=|N(u)|$.
We use ${\Delta_{S}}^{(i)}$ to denote the $i$th largest degree among vertices in $S$.
The maximum degree among all vertices in $L$, $R$ and $L\cup R$ are denoted by $\Delta_L$, $\Delta_R$ and $\Delta$, respectively.

For any vertex $u\in V$ 
and a positive integer $k$, we use $N^k(u)$ to denote the set of vertices with distance exactly $k$ to $u$ in $G$. 
The vertices in $N^k(u)$ are called $k$-hop neighbours of $u$. In particular, 
$N^1(u)=N(u)$. 

Lastly, given two vertex sets $X\subseteq L$ and $Y\subseteq R$, let $G[X\cup Y]$ be the subgraph induced by $X\cup Y$. 
We define the diameter of $G$ as the maximum distance among all pairs of vertices. { The complement graph of the bipartite graph $G$ is a bipartite graph $\overline{G}=(L\cup R,E')$ where $E'=\{(u,v)\mid (u,v)\in L\times R\setminus E\}$.}

\subsection{Problem Formulation}
\begin{definition}[$s$-biplex]
Given a positive integer $s$, a bipartite graph $B=(X\cup Y,E)$ is an $s$-biplex if $|Y|-|N_Y(u)|\leq s$ for each $ u\in X$ and $|X|-|N_X(v)|\leq s$ for each $ v\in Y$.
\end{definition}

An $s$-biplex is hereditary -- any induced subgraph of the $s$-biplex remains an $s$-biplex \citep{yu2021efficient}. The hereditary property motivates the following class of $s$-biplexes.

\begin{definition}[Maximal $s$-biplex] 
Given a bipartite graph $G=(L\cup R,E)$, a positive integer $s$, and a vertex set $B\subseteq L\cup R$, the induced subgraph $G[B]$ is a maximal $s$-biplex in $G$ if $G[B]$ is an $s$-biplex, and for each $B'\supseteq B$, $G[B']$ is not an $s$-biplex.
\end{definition}



 
As mentioned, in most data engineering tasks, the focus is only on large $s$-biplexes. 
However, a maximal $s$-biplex is not always large, e.g. $|X|$ could be as small as 0
in an $s$-plex $B=(X\cup Y,  E')$. 
So we introduce two positive integers $\theta_L$ and $\theta_R $ for the target $s$-biplex $B$, such that $|X|\geq \theta_L$ and $|Y|\geq \theta_R$. These two integers, namely the \textit{size lower bounds}, help to filter out extreme cases, that is, the number of vertices on one side is extremely larger than that at the other side. 
Moreover, because disconnected graphs are often deemed not interesting in near-clique subgraphs \citep{asratian1998bipartite,wang2021efficient}, we focus on the case where both $\theta_L$ and $\theta_R$ are larger than $2s+1$, where the graph is proven to be connected with a diameter bounded by 3 in Lemma 1 of \citep{yu2023maximum}. 

We now introduce the {\em top-$k$ $s$-biplex search (TBS) problem}, followed by its computational complexity analysis.

\begin{problem}[Top-$k$ $s$-biplex search (TBS) problem] \label{problem: definition}
Given a bipartite graph $G=(V=L\cup R,E)$, and four positive integers $k>0$, $s>0$, $\theta_L \geq 2s+1$ and $\theta_R \geq 2s+1$, find the largest $k$ subsets of $V$ such that each subset $S$ satisfies 
\begin{enumerate}
    \item the subgraph induced by $S$ is a maximal $s$-biplex, and
    \item $|S\cap L|\geq \theta_L$, $|S\cap R|\geq \theta_R$.
\end{enumerate}
\end{problem}


When $k=1$, the TBS problem is equivalent to finding the maximum $s$-biplex (with size bound) from a given graph. We show that in this case, the TBS is already NP-hard.
\begin{theorem}
The {\em TBS} problem is NP-hard.
\end{theorem}

\begin{proof}
Our proof is built on a reduction from the {\it constrained minimum vertex cover (CMVC) problem on bipartite graphs} \citep{chen2003constrained}, which is NP-hard, to the {\em TBS} problem when $k=1$.

Our reduction consists of two steps. 

The first step reduces the CMVC problem to the size-constrained maximum biclique problem in polynomial time. 
Given a bipartite graph $G(L\cup R,  E)$ and two constraints $k_L$ and $k_R$, the CMVC problem aims to find a minimum subset of vertices $T\subseteq V$ such that (1) every edge has at least one endpoint in the subset, (2) the subset contains at most $k_L$ vertices in $L$ and at most $k_R$ vertices in $R$.
By the classical reduction from vertex cover to independent set \citep{kleinberg2006algorithm}, if $T$ is a minimum vertex cover (with at most $k_L$ vertices in $L$ and at most $k_R$ vertices in $R$), then $V\setminus T$ is a maximum independent set (with at least $|L|-k_L$ vertices in $L$ and at least $|R|-k_R$ vertices in $R$). 
And then in the complement graph $\overline{G}$ of $G$, $\overline{G}[V\setminus T]$ is the maximum biclique in $\overline{G}$ with at least $|L|-k_L$ vertices in $L$ and at least $|R|-k_R$ vertices in $R$.
In summary, $T\subseteq V$ is a constrained minimum vertex cover in $G$ if and only if $V\setminus T$ is a size-constrained maximum biclique in {$\overline{G}$}.

The second step reduces the size-constrained maximum biclique problem to the {\em TBS} problem in polynomial time. We define the decision version of these two problems:\newline
(i) BICLIQUE: given a bipartite graph $G=(L\cup R,E)$ and three positive integers $k_L$, $k_R$, and $\alpha\geq k_L+k_R$, is there a biclique of size $\alpha$, with at least $k_L$ vertices in left side, at least $k_R$ vertices in right side?\newline
(ii) BISPLEX: given a bipartite graph $G=(L\cup R,E)$ and four positive integers $s$, $\theta_L$, $\theta_R$, and $\alpha'\geq \theta_L+\theta_R$, is there an $s$-biplex of size $\alpha'$ with at least $\theta_L$ vertices in left side, at least $\theta_R$ vertices in right side?


\begin{figure}[h]
    \centering
    \includegraphics[width=0.5\linewidth]{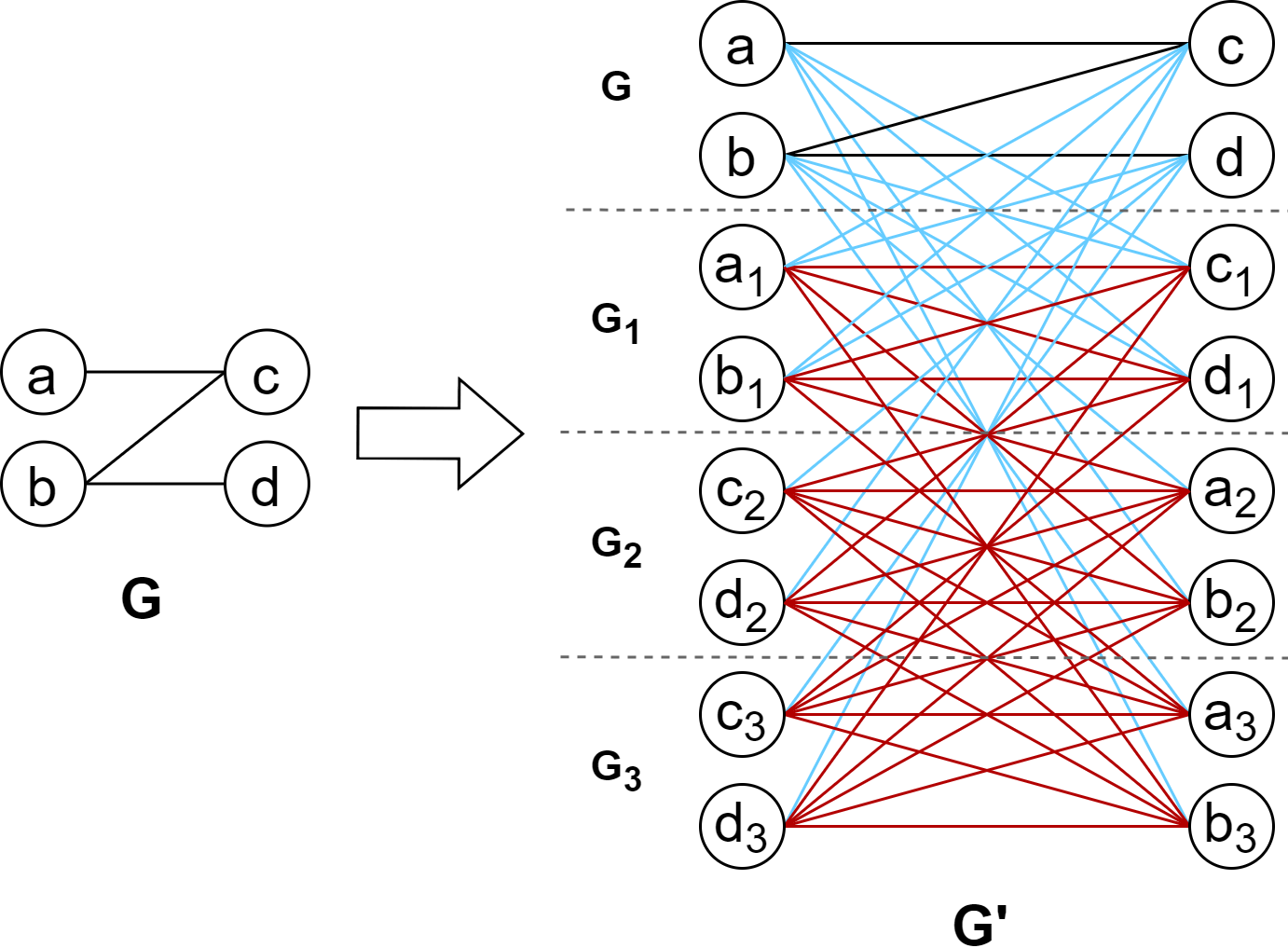}
    \caption{An example of the graph construction process.}
    \label{fig:construction}
\end{figure}

Given an input of BICLIQUE $\langle G,k_L,k_R,\alpha\rangle$, the crucial step in our proof is to construct a graph $G'$ that is the input of the BISPLEX. 
{
Firstly, we build $2s-1$ graphs: $G_i=(L_i\cup R_i,E_i)$ ($i=1,\ldots,2s-1$), where $L_i=L,R_i=R$  for $i\leq s-1$, $L_i=R,R_i=L$ for $i \geq s$, and $E_i=\varnothing$ for $i=1,\ldots,2s-1$. Now, each vertex $u_k\in L$ (resp. $v_k\in R$) has a copy in every $L_i$ with $i=s,\ldots,2s-1$ (resp. $R_i$ with $i=s,\ldots,2s-1$). 
Next, we connect each vertex in $L$ (resp. $R$) to all vertices in $R_i$ (resp.$L_i$) where $i=1,\ldots,2s-1$ except its copies. And we connect each vertex in $L_i$ (resp. $R_i$) to all vertices in $R_i$ (resp. $L_i$) with $i=1,\ldots,2s-1$ except its copies.}
We denote these connections as $E'$.
The constructed graph is defined as $G'=(L\cup L'\cup R\cup R',E\cup E')$ where $L'= \bigcup_{i=1}^{2s-1}L_i$, $R'=\bigcup_{i=1}^{2s-1}R_i$.
An example of this construction is shown in Figure~\ref{fig:construction}.

Define $\theta_L:=k_L+(2s-1)|L|$ 
, $\theta_R:=k_R+(2s-1)|R|$ and $\alpha':=\alpha+k_L+k_R+(2s-1)n$,where $n = |L\cup R|$.
In the following, we prove that there is a solution for BISPLEX on inputs $\langle G',s,\theta_L,\theta_R,\alpha'\rangle$ 
if and only if there is a solution for BICLIQUE on inputs $\langle G,k_L,k_R,\alpha\rangle$.

$\Leftarrow$.
Given a solution of BICLIQUE induced by $B=B_L\cup B_R$ in $G$.
We will prove that $B'=B_L\cup L'\cup B_R\cup R'$ is a solution for BISPLEX in $G'$. {
According to the construction of $G'$, each vertex $u\in B_L$ (resp. $v\in B_R$) 
is connected to all vertices in $B_R\cup R'$ (resp. $B_L\cup L'$) except its $s$ copies.
And each vertex $u\in L'$ (resp. $v\in R'$) has exactly $s$ anti-neighbours in $B_R\cup R'$ (resp. $B_L\cup L'$).} Therefore, $B'$ is an $s$-biplex.
Moreover, $B'$ contains exactly $\alpha+(2s-1)n$ vertices with $|B_L|+(2s-1)|L|\geq k_L+(2s-1)|L|$ vertices in left side, and $|B_R|+(2s-1)|R|\geq k_R+(2s-1)|R|$ vertices in right side. Thus, $B'$ is a solution for BISPLEX on inputs $\langle G',s,\theta_L,\theta_R,\alpha'\rangle$.

$\Rightarrow$.
Given that there is no solution $B$ for BICLIQUE in $G$.
Assume, for the sake of contradiction, that there exists a solution $P=P_L\cup P_R$ for BISPLEX on inputs $\langle G',s,\theta_L,\theta_R,\alpha'\rangle$.

We claim that given $P$, there always exists an $s$-biplex induced by $P'=P'_L\cup P'_R$ such that 
$P'_L\supset L'$, $P'_R\supset R'$, $|P'_L|\geq \theta_L$, $|P'_R|\geq \theta_R$, and $|P'|=\alpha'=\alpha+(2s-1)n$.
When $P_L\supset L'$ and $P_R\supset R'$, the claim clearly follows.
When $L'\setminus P\neq \varnothing$ or $R'\setminus P\neq \varnothing$, the following proves the claim by constructing $P'$: We initialise $P'=P$.
For a vertex $u' \in L'\setminus P'$, w.l.o.g., assume $u'$ is the copy of a vertex $u\in L\cup R$ (such vertex $u$ must exist according to the construction of $G'$). 
If $u\in P'_R$ (resp. $u\in P'_L$), then we remove from $P'$ all vertices in $\overline{N}_{P'\cap L}(u)$ (resp. $\overline{N}_{P'\cap R}(u)$), and add  all vertices in $\overline{N}_{L'\setminus P'}(u)$ (resp. $\overline{N}_{R'\setminus P'}(u)$) into $P'$. $P'$ remains an $s$-biplex because the hereditary property and vertex in $\overline{N}_{L'\setminus P'}(u)$ (resp. $\overline{N}_{R'\setminus P'}(u)$) is connected to all vertices except its $s$ copies.
Since $|\overline{N}_{L'\setminus P'}(u)|-|\overline{N}_{P'\cap L}(u)|=s-|\overline{N}_{L'\cap P'}(u)|-|\overline{N}_{P'\cap L}(u)|\geq 0$ (resp. $|\overline{N}_{R'\setminus P'}(u)|-|\overline{N}_{P'\cap R}(u)|=s-|\overline{N}_{R'\cap P'}(u)|-|\overline{N}_{P'\cap R}(u)|\geq 0$), the total number of vertices in the left side (resp. right side) does not decrease.
If $u\notin P'$, then we directly add $u'$ into $P'$. $P'$ remains an $s$-biplex because $u'$ has at most $s$ anti-neighbours in $G'$.
We perform the operations above for all vertices in $L'\setminus P'$ and $R'\setminus P'$.
After performing the operations, all vertices in $L'$ and $R'$ are included in $P'$. 
Next, we refine $P'$ by randomly removing vertices from $P'_L\cap L$ and $P'_R\cap R$ such that
$|P'_L|\geq \theta_L$, $|P'_R|\geq \theta_R$, and $|P'|=\alpha+2sn$.
According to the hereditary property, the refined $P'$ is an $s$-biplex.

There are at least $k_L$ vertices in $P'_L\cap L$ and  $k_R$ vertices in $P'_R\cap R$ since $\theta_L = k_L + (2s-1)|L|$ and $\theta_R = k_R + (2s-1)|R|$. For each vertex in $P'_L\cap L$ (resp. $P'_R\cap R$), it has $s$ anti-neighbours in $R'$ (resp. $L'$). 
In other words, each vertex in $P'_L\cap L$ (resp. $P'_R\cap R$) is connected to every vertex in $P'_R\cap R$ (resp. $P'_L\cap L$). Therefore, $B=(P'_L \cap L) \cup(P'_R \cap R) $ forms a biclique where $|P'_L \cap L|\geq k_L$ ,  $|P'_R \cap R|\geq k_R$ and $|B|= \alpha$, which contradicts our initial assumption. Hence, no solution exists for BISPLEX on inputs $\langle G',s,\theta_L,\theta_R,\alpha'\rangle$.

In summary, there is a polynomial-time reduction from CMVC, which is NP-hard \citep{chen2003constrained}, to BISPLEX, implying that 
the {\em TBS} problem is also NP-hard.
\end{proof}

\subsection{Related Works}
The {\em TBS} problem  is related to the problems of listing cliques/plexes on non-bipartite graphs and bipartite graphs.


\noindent {\bf Listing cliques and $s$-plexes on non-bipartite graphs.}
Cliques are known as a complete graph where each pair of vertices is connected by an edge.
$s$-Plex is a relaxation of clique by allowing at most $s$ non-adjacent vertices for each vertex.
Intensive efforts have been devoted to searching cliques \citep{bron1973algorithm,jin2022fast} and s-plexes \citep{conte2018d2k,zhou2020enumerating,wang2022listing}.
Branching-and-bound algorithms stand out due to their high efficiency.
\citep{bron1973algorithm} proposed the Bron-Kerbosch algorithm to list maximal cliques.  
\citep{conte2018d2k} proposed an improved branching algorithm $D2K$ using graph decomposition. $D2K$ can list maximal $s$-plexes in $O^*(2^{d_1\Delta})$ time, where $d_1$ is the 1-hop degeneracy.
On top of the Bron-Kerbosch algorithm, \citep{zhou2020enumerating} devised an improved algorithm with a pivot heuristic. Their algorithm runs in $O^*(\zeta_s^{d_1\Delta})$ time, while $\zeta_s$ is a constant larger than our $\gamma_s$.  Next, \citep{dai2022scaling} improve the time complexity to $O^*(\zeta_s^n)$.
Independently, \citep{wang2022listing} proposed the ListPlex algorithm, which further reduces the time complexity to $O^*(\zeta_s^{d_1})$.  However, using {\em ListPlex} to search $s$-biplexes needs  $O^*(\zeta_s^{\Delta d_1})$ time due to the additional time for considering vertices in opposite side.
More recently, \citep{wu2024efficient} extended the task of enumerating maximal $s$-plexes to temporal graphs.

\noindent {\bf Listing bicliques and $s$-biplexes on bipartite graphs.}
Cliques in bipartite graphs are called Bicliques, where each vertex is adjacent to all vertices on the opposite side.
$s$-Plex in bipartite graphs are called $s$-biplexes where 
each vertex is not adjacent to at most $s$ vertices on the opposite side.
Existing research on listing bicliques has explored various approaches, including searching for maximal bicliques with the most edges \citep{sozdinler2018finding,lyu2020maximum,shaham2016finding}, finding maximal bicliques with the most vertices \citep{garey1979computers}, and identifying maximal balanced bicliques where the left-side size equals the right-side size \citep{zhou2018towards,zhou2017combining,wang2018new}.
Compared to approximation algorithms \citep{wang2018new}, branch-and-bound algorithms \citep{zhou2018towards} are preferable since they guarantee the optimal solution.

For biplexes, \citep{luo2022maximum} proposed {\em MBS-Core} to find the maximal $s$-biplex with most edges in $O^*(2^{n})$ time. Their main techniques are a progressive bounding framework {\em Core} and a branch-and-bound enumeration algorithm {\em MBS}.
On top of this work, \citep{yu2023maximum} proposed FastBB to list $k$ maximal $s$-biplexes with most edges in $O^*(\gamma_s^{\Delta^3})$ time, where $\Delta$ is the maximum vertex degree and $\gamma_s<2$ is a constant. 
This paper further improves the time complexity to $O^*(\gamma_s^{d_2})$, where $d_2<\Delta^2$ is the largest 2-hop degeneracy. 


\section{A Basic Branching algorithm}
In this section, we propose  a branching algorithm  for the
{\em top-$k$ $s$-biplex search (TBS)} problem
and analyse its time complexity. The pseudo-code is given in Algorithm~\ref{alg:mvbp_topk}.

\begin{algorithm}[ht!]
\caption{MVBP} \label{alg:mvbp_topk}
 \begin{algorithmic}[1]
 \item[INPUT:] a graph $G(L\cup R,E)$, four integers $s,\theta_L,\theta_R$, and $k$.
 \item[OUTPUT:] top-$k$ maximal $s$-biplexes with most vertices in $\mathcal{S}$ or all $s$-biplexes if $|\mathcal{S}|\leq k$.
\State initialise $\mathcal{S}$ with $k$ empty sets as $k$ $s$-biplexes at hand, { $lb_L\leftarrow \theta_L,lb_R\leftarrow \theta_R$}
\State run {\emph MVBPF$(G,s,\emptyset,L\cup R,\emptyset)$}  and update $\mathcal{S}$
\Statex {\bf Procedure} {\emph MVBPF$(G,s,P,C,X)$} 
 \State $C\gets \{v: v\in C  \text{ and } G[\{v\}\cup{P}] \text{is an $s$-biplex} \}$
\State $X \gets \{v: v\in X  \text{ and } G[\{v\}\cup{P}] \text{ is an $s$-biplex} \}$ 
\Statex //* Pruning
\If{$X=\varnothing$ and $C=\varnothing$} update $\mathcal{S}$ and {\bf return}  \Comment{PR1}
\EndIf
\If{ $\forall S\in \mathcal{S}$ it holds that $|P\cup C|\leq |S|$} {\bf return} \Comment{PR2}
\EndIf
 \If{{ there exists $u\in X$ such that $|\overline{N}_{P\cup C}(u)|\leq s$, and for each $v\in P\cup C$ with $|\overline{N}_{P\cup C}(v)|\geq s$ it holds that $(u,v)\in E$}}
 {\bf return}   \Comment{PR3}  \label{line:maximal exclusion}
 \EndIf

\If{for each $u\in P\cup C$, it holds that $|\overline{N}_{P\cup C}(u)|\leq  s$} 
update $\mathcal{S}$ and
 {\bf return}  
 \Comment{PR4}
 \EndIf

 \If{$|P_L\cup C_L|<lb_L$ or $|P_R\cup C_R|<lb_R$} {\bf return} \Comment{PR5}
 \EndIf 
  \Statex //* Branching
 \If{$\exists u\in P$, $|\overline{N}_{P\cup C}(u)|\geq s+1$} \Comment{BR1} \label{line:uinP}
 \State 
  select pivot $u_p\leftarrow arg \max_{v\in P,|\overline{N}_{P\cup C}(v)|\geq s+1}|\overline{N}_{P\cup C}(v)|$ 
 \State let $u_1,\ldots,u_b$ be vertices in $\overline{N}_C(u_p)$ \label{line:ordering}
 \State $s'\leftarrow{s-|\overline{N}_P(u_p)|}$
 \State run \emph{MVBPF}$(G,s,P,C\backslash \{u_1\},X\cup \{u_1\})$
 \For{$i=2,\ldots,s'$}
 \State \emph{MVBPF}$(G,s,P\cup \{u_1,\ldots,u_{i-1}\},  C\backslash \{u_1,\ldots,u_{i}\}, X\cup \{u_i\})$
 \EndFor
  \State \emph{MVBPF}$(G,s,P\cup \{u_1,\ldots,u_{s'}\}, C\backslash \{u_1,\ldots,u_b\}, ${ $X\cup\{u_{s'+1},\ldots,u_b\})$}
 \Else  \Comment{BR2} \label{line:uinC}
 \State 
 select pivot $u_p\leftarrow arg \max_{v\in C}|\overline{N}_{P\cup C}(v)|$

  \State \emph{MVBPF}$(G,s,P\cup \{u_p\},C\backslash \{u_p\},X)$
 \State apply PR6 and update $P$ \Comment{PR6}
 \State \emph{MVBPF}$(G, s, P, C\backslash \{u_p\}, X\cup \{u_p\})$
 \EndIf

 \end{algorithmic}
 \end{algorithm}

\subsection{The Basic Algorithm}
The proposed algorithm {\emph MVBPF} is a depth-first branching algorithm following the setting of the basic Bron-Kerbosch algorithm \citep{bron1973algorithm}.
We employ three disjoint sets  $P$, $C$ and $X$ in each branch.
$P=P_L\cup P_R$ denotes the \textit{partial solution set}, which contains vertices that must be included in every subsequent branch.
$C=C_L\cup C_R$ denotes the \textit{candidate set} which contains vertices that can be included in $P$.
$X=X_L\cup X_R$ denotes the \textit{exclusion set} whose vertices must not be included in any subsequent branches. 
In the following of the paper, the subscripts $L$ and $R$ identify the vertices on the left side and the right side, respectively.

{
Note that a vertex in the exclusion set $X$ can form a larger $s$-biplex with $P$ and has been visited in previous branches. 
When $X=\varnothing$, then the current $s$-biplex $P$ is not visited before (line 5).
Therefore, we  maintain $X$ in order to avoid revisiting the $s$-biplex (line 4).
}

In each branch, it works on the following subproblem. Given a partial solution set $P=P_L\cup P_R$, a candidate set $C =C_L\cup C_R$ and an exclusion set $X=X_L\cup X_R$, enumerate the vertex set $ P\subseteq B \subseteq P\cup C$ such that any vertex $u\in X$ cannot form an $s$-biplex with $B$. 
When such a vertex set $B$ is found, it means that a maximal $s$-biplex is found. 
Because we are only interested in the $k$ largest solutions, we use $\mathcal{S}$ to keep track of the largest $k$ solutions found so far.
In fact, this branching search scheme has been widely used to enumerate different maximal subgraphs \citep{bron1973algorithm,tomita2006worst,conte2018d2k,zhou2020enumerating}.

In each branch, {\emph MVBPF} consists of a pruning component and a branching component. The pruning component tries to reduce the candidate set $C$ or halt the search. The branching component test different possibilities that a vertex of $C$ is in the maximal $s$-biplex or not.
We first present six basic pruning principles used in the pruning component. 

\begin{itemize}[leftmargin=*] 
    \item {\bf PR1.} If both the exclusion set $X$ and the candidate set $C$ are empty, then no more vertices can be further included in $P$. In this case, $P$ has been a maximal $s$-biplex.
    
    \item {\bf PR2.} If the union of partial set $P$ and candidate set $C$ is not larger than the current $k$-th $s$-biplex, then the top-$k$ maximal $s$-biplexes will not appear in all subsequent branches. In this case, all subsequent branches can be pruned. 
    
    \item {\bf PR3.} If a vertex $u$ in exclusion set $X$ has less than $s$ anti-neighbours in $P\cup C$, and links to all vertices in $P\cup C$ that have no less than $s$ anti-neighbours in $P\cup C$, then the vertex $u$ can be included by every $s$-biplexes found in $P\cup C$. 
    This means that all these $s$-biplexes can not be maximal, and, therefore, can not be one of the top-$k$ maximal $s$-biplexes. In this case, all subsequent branches can be pruned.

    \item {\bf PR4.} If $P\cup C$ induces an $s$-biplex, then the subsequent branches cannot get maximal $s$-biplexes. 
    Furthermore, not returning via PR3 ensures $P\cup C$ induces a maximal $s$-biplex.

    \item {\bf PR5.} Let $lb_L$ (resp. $lb_R$) denote the lower bound of the number of vertices in $L\cap P$ (resp. $R\cap P$).
    If the total number of vertices $|P_L\cup C_L|$ (resp.$|P_R\cup C_R|$ ) has already been smaller than the lower bound $lb_L$ (resp. $lb_R$), then subsequent branches cannot get $s$-biplexes that satisfy lower bounds, and thus can be pruned.
\end{itemize}

After the pruning above, we can guarantee that there exists a vertex in $P\cup C$ with at least $s+1$ anti-neighbours in $G[P\cup C]$. 

\begin{itemize}[leftmargin=*]
    \item {\bf PR6.} Given that $|\overline{N}_{P\cup C}(u)| \leq s$  for each $u\in P$, if the pivot $u_p\leftarrow arg \max_{v\in C}|\overline{N}_{P\cup C}(v)|$ has no anti-neighbour in $P$ and has exactly $s+1$ anti-neighbours in $C$,  then every anti-neighbour of $u_p$ has at most $s+1$ anti-neighbours in $P\cup C$, and has at most $s$ anti-neighbours in $P\cup C\setminus \{u_p\}$. Therefore, each maximal $s$-biplex that excludes $u_p$ must includes $\overline{N}_C(u_p)$. 
\end{itemize}

The branching component consists of two branching rules.
\begin{itemize} [leftmargin=*] 
    \item {\bf BR1.} If there exists a vertex in $P$ such that $|\overline{N}_{P\cup C}(u_p)|\geq s+1$, then we select pivot $u_p\leftarrow$
    { $arg \max_{v\in P, |\overline{N}_{P\cup C}(v)|\geq s+1}|\overline{N}_{P\cup C}(v)|$} and generate $s'+1$ branches where $s'=s-|\overline{N}_P(u_p)|$. There are at most $s'$ vertices in  $\overline{N}_C(u_p)$ that can be included in $P$.  
    Let $\{u_1,u_2\ldots,u_b\}$ be an arbitrary order of vertices in $\overline{N}_C(u_p)$. { Note that $|\overline{N}_C(u_p)|=|\overline{N}_{P\cup C}(u_p)|-|\overline{N}_P(u_p)|\geq s'+1$.} we generate $s'+1$ branches:\par
    (1) In the first branch, 
    move $u_1$ from $C$ to $X$;\par
    (2) For $2\leq i\leq s'$, move $\{u_1,\ldots,u_{i-1}\}$ from $C$ to $P$, and move $u_i$ from $C$ to $X$.\par
    (3) In the last branch, move $\{u_1,\ldots,u_{s'}\}$ from $C$ to $P$ and $\{u_{s'+1},\ldots,u_b\}$ from $C$ to $X$ (i.e., all vertices in $\overline{N}_C(u_p)$ are removed from $C$ ).\par

    \item {\bf BR2.} If for each $u\in P$ we have $ |\overline{N}_{P\cup C}(u)| \leq s$, then we select a pivot vertex $u_p$ with the maximum anti-neighbours in $G[P\cup C]$ and generate two branches: move $u_p$ from $C$ to $X$, move $u_p$ from $C$ to $P$.
\end{itemize}


\begin{figure}[ht!]
\centering 
\includegraphics[width=0.6\linewidth]{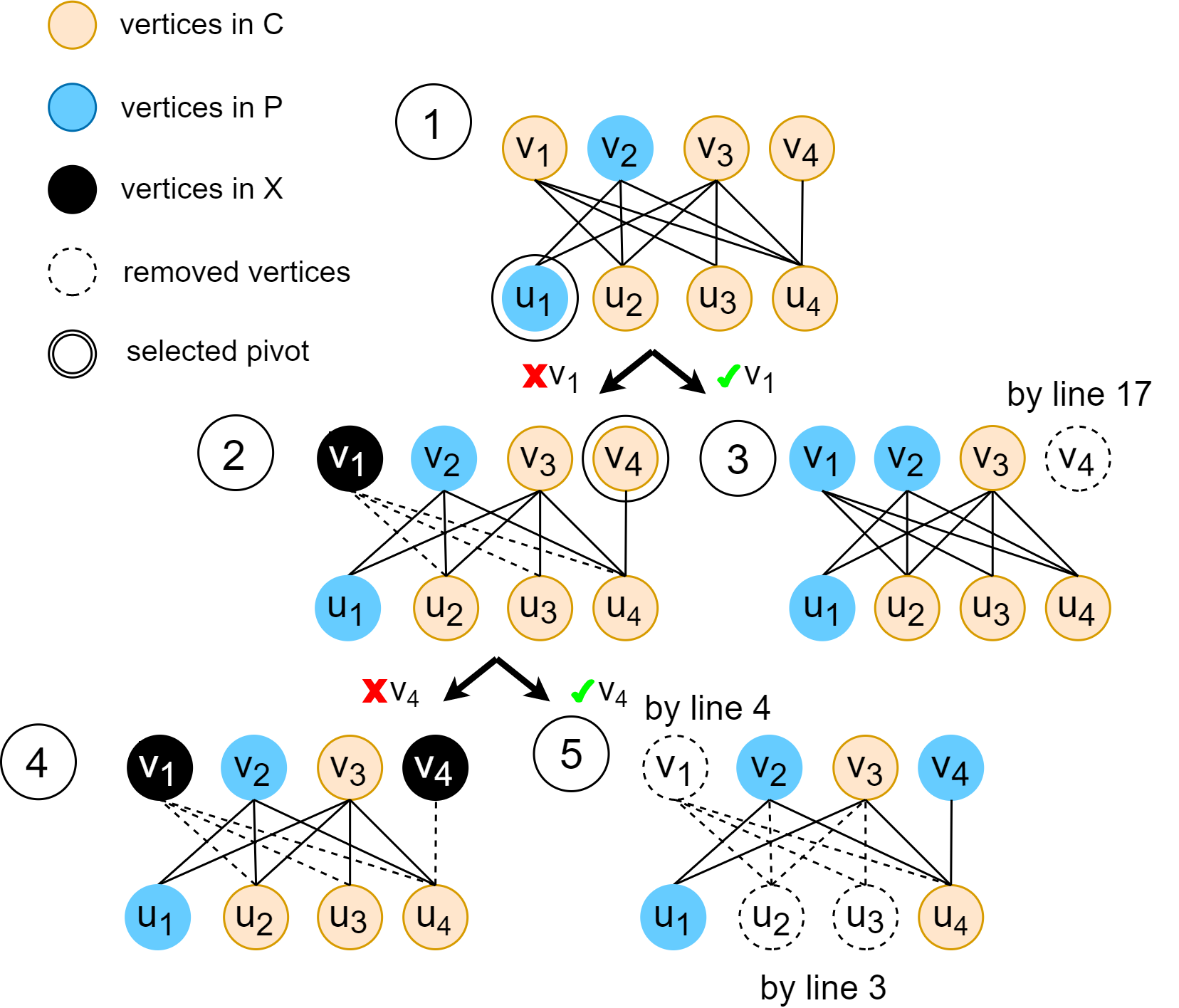}
\caption{An example of Branching where $\theta_L=\theta_R=0,s=1$. Blue, yellow and black vertices represent vertices in $P,C$ and $X$.
In \textcircled{1}, pivot is $u_1\in P$ (line~\ref{line:uinP}), and $\overline{N}_C(u_1)=\{v_1,v_4\},\overline{N}_P(u_1)=\emptyset $, so we create $s'+1=s-|\overline{N}_P(u_1)|+1=2$ branches (BR1), i.e., \textcircled{2} (line 14) and \textcircled{3} (line 17). 
In \textcircled{2}, the vertex $v_4\in C$ is selected to be the pivot (line~\ref{line:uinC}). We then creates 2 branches (BR2), i.e.,
\textcircled{4} by moving $v_4$ from $C$ to $X$ (line 22) and \textcircled{5} by moving $v_4$ from $C$ to $P$ (line 20). 
In \textcircled{3} and \textcircled{5}, $P\cup C$ has been a $1$-biplex, so we stop the branching (line 8). In \textcircled{4}, vertex $v_1$ in $X$ can be added to form a larger $1$-biplex, so we stop the branching by PR3 (line 7).
}
\label{fig:branchingStrategy}
\end{figure}

The branching above can cover all potential branches.
Firstly, for each vertex in $C$, the procedure branches to move it to $X$ or to $P$ which covers all possible placement of the vertex. 
Next, for each vertex in $P$ with more than $s$ anti-neighbours, the procedure cover the branches that exclude $u_i$ in (1) and (2), and cover the brances that include $u_i$ in (2) and (3). 
An example of the branching strategy is shown in Fig.~\ref{fig:branchingStrategy}.
\subsection{Time Complexity of MVBP}\label{BaseTimeComplexity}
In this section, we analyse the worst-case running time of Algorithm~\ref{alg:mvbp_topk}. The final time complexity results are shown in the following theorem.

\begin{theorem} \label{thm:timemvbp}
Given a bipartite graph $G=(L\cup R,E)$, MVBP finds top-$k$ maximal $s$-biplexes in time
$O(n\Delta\gamma_s^{n})$, where $\Delta$ is the maximum degree of $G$, and $\gamma_s$ is the largest real root of $x^{s+4} - 2x^{s+3} + x^2 -x +1 =0$, which is strictly less than 2. { For $s=0,1,2,3$, $\gamma_s=1.465,1.754,1.888, 1.947$, respectively.}
\end{theorem}
{Because MVBP is a typical exponential branching algorithm, we apply the standard techniques of bounding the complexity of a branching algorithm. Specifically, the complexity of a branching algorithm is  bounded by its number of branches in the search tree multiplied by the time of generating the branches \footnote{Interested readers can refer to Chapter 2 in the monograph \citep{fomin2013exact} for details}.} 
\begin{proof}
Denote $T(|C|)$ as the total number of leaf nodes in the branching tree. Note that $|C|=|L\cup R|=n$ (the set $C$ will be reduced in our improved algorithm). 

For $u_p\in P$ (in Line~\ref{line:uinP}), there are $s'+1$ branches. In particular, in the $i$-th branch where $i=1,\ldots,s'$, $i$ vertices are removed from $C$ for each branch. And for the last branch, there are $b$ vertices removed from $C$. Therefore, $T(n)\leq \sum_{i=1}^{s'}T(n-i)+T(n-b)$.
{ Given $s\geq s'$ and $b=|\overline{N}_C(u_p)|=|\overline{N}_{P\cup C}(u_p)|-|\overline{N}_{P}(u_p)|\geq s+1-|\overline{N}_{P}(u_p)|= s'+1$},
we have $T(n)\leq \sum_{i=1}^{s'+1}T(n-i)\leq \sum_{i=1}^{s+1}T(n-i)$.

For $u_p\in C$, there are two branches. In the first branch, $u_p$ is removed from $C$. The second branch moves $u_p$ from $C$ to $P$. Notice that $u_p\in P$ further creates $s'+1$ branches as mentioned above. Therefore, there are at most $s'+2$ branches on $u_p$. And we have $T(n)\leq \sum_{i=1}^{s'+1}T(n-i)+T(n-b-1)$, where $s'=s-|\overline{N}_P(u_p)|$ and $b=|\overline{N}_C(u_p)|$.

The following discuss the bounds of $s'$ and $b$. 
There are three cases.

\noindent
\textbf{Case 1:} $u_p\in C$ has at least one anti-neighbour in $P$.
In this case, $|\overline{N}_P(u_p)|\geq 1$ and thus 
$s'=s-|\overline{N}_P(u_p)|\leq s-1$. 
Notice that 
$b=|\overline{N}_C(u_p)|\geq s'+1$
In summary, $T(n)\leq \sum_{i=1}^{s'+2}T(n-i)
\leq \sum_{i=1}^{s+1}T(n-i)$.

\noindent 
\textbf{Case 2:} $u_p\in C$ has no anti-neighbour in $P$, and has at least $s+2$ anti-neighbours in $C$.
In this case, $s'=s-|\overline{N}_P(u_p)|=s$ and $b=|\overline{N}_C(u_p)|\geq s+2$. Therefore $T(n)\leq \sum_{i=1}^{s+1}T(n-i)+T(n-s-3)$.

\noindent
\textbf{Case 3:} $u_p\in C$ has no anti-neighbour in $P$ and has exactly $s+1$ anti-neighbours in $C$.
In this case, $s'=s-|\overline{N}_P(u_p)|=s$ and $b=|\overline{N}_C(u_p)|=s+1$. Therefore $T(n)\leq \sum_{i=1}^{s+2}T(n-i)$.
Moreover, by PR6, we can reduce $s+2$ vertices in the branch, excluding $u_p$, by including all anti-neighbours of $u_p$.
In summary, we have $T(n)\leq \sum_{i=2}^{s+2}T(n-i)+T(n-s-2)$.

Observe that the number of branches in Case $u_p\in P$ and Case 1 are bounded by that in Case 2.
In the following, we compare the number of branches in Case 2 and Case 3.

In order to bound the number of branches of case 2, we get $T(n)\leq \gamma_s^n$ where $\gamma_s$ is the largest real root of the equation $x^n - x^{n-1}-x^{n-2}- \ldots - x^{n-s-1} - x^{n-s-3}=0$, that is, the largest real root of $x^{s+4}-2x^{s+3}+x^{2}-x+1=0$.

By applying the same techniques in case 3, we get
$T(n)\leq \kappa_s^n$, where $\kappa_s$ is the largest real root of the equation $x^n - x^{n-2}-x^{n-3}- \ldots - x^{n-s-2} - x^{n-s-2}=0$, that is, the largest real root of $x^{s+3}-x^{s+2}-x^{s+1}-x+2=0$.

Consider the function $f(x)=(x^n - x^{n-2}-x^{n-3}- \ldots - x^{n-s-2} - x^{n-s-2}) - (x^n - x^{n-1}-x^{n-2}- \ldots - x^{n-s-1} - x^{n-s-3})=x^{n-s-3}+x^{n-1}-2x^{n-s-2}.$ 
Notice that $f(x)=x^{n-s-3}*(x^{s+2}+1-2x)\geq x^{n-s-3}*(x^{2}-2x+1)\geq 0$ for $s \geq 0$ and $x>1$.
And $x=1$ is a root for both $x^{s+4}-2x^{s+3}+x^{2}-x+1=0$ and $x^{s+3}-x^{s+2}-x^{s+1}-x+2=0$. we conclude that $\gamma_s\geq \kappa_s$.
Therefore, $T(n)\leq \gamma_s^n$ for all cases.

Due to the fact that each branch costs at most $O(n\Delta)$ time, we conclude that the total time complexity of MVBP is bounded by $O( n\Delta\gamma_s^n)$. 
\end{proof}

\section{An Improved Algorithm}
\label{section:techniques}
In this section, we provide an improved algorithm, as shown in Algorithm~\ref{alg:fastmvbp}, by employing three techniques: 2-hop decomposition (line~\ref{line: neighbourhood begin}-\ref{line: decomposition end}), reduction with single-side bounds (line~\ref{line: reduction_fastmvbp}), and progressive search (line~\ref{line: PB}-\ref{line: PB end}). { The value ${\Delta_{R}}^{(\theta_R)}$ (resp. ${\Delta_{L}}^{(\theta_L)}$) denote the $\theta_R$th (resp. $\theta_L$th) largest degree of vertex among vertices in $R$ (resp. $L$). An illustration of key components in Algorithm~\ref{alg:fastmvbp} is shown in Figure~\ref{fig:fastmvbp framework}}.

In the following subsections, we describe the three techniques one-by-one.

\begin{algorithm}[ht!]
\caption{FastMVBP}
\label{alg:fastmvbp}
\begin{algorithmic}[1]
\item[INPUT:] a graph $G(L\cup R,E)$, four integer $s,\theta_L,\theta_R$, and $k$.
\item[OUTPUT:] all biplexes in $\mathcal{S}$
\State initialise $\mathcal{S}$ with $k$ empty sets as top-$k$ maximal $s$-biplexes
\State let $lb$ denote the size of smallest set in $\mathcal{S}$
\State update $\mathcal{S}$ and $lb$ during the running of {\emph MVBPF} 
\State  $ub_L\leftarrow \Delta^{(\theta_R)}_R+s$,  $lb_R\leftarrow \theta_R$, $ub_R\leftarrow \Delta^{(\theta_L)}_L+s$
\State { $\Phi \leftarrow \max(ub_L/2,\theta_L)$}, $lb_L \leftarrow \Phi$ 
\State run $computeDegeneracy$ and get
$\eta=\{v_1,\ldots,v_n\}$
\While{$lb_L > \theta_L$}  \label{line: PB}
\State $lb_L\leftarrow \Phi$
\State {\emph MVBPD$(G,s,\eta)$}
\If{{ there are $k$ non-empty sets in $\mathcal{S}$}}
{  $ub_L \leftarrow lb_L$} and break
\EndIf
\State { $ub_L\leftarrow \Phi$,}$\Phi\leftarrow \max(\frac{\Phi}{2},\theta_L)$
\EndWhile

\State $lb_L\leftarrow \theta_L$, $lb_R\leftarrow \max(lb+1-ub_L,\theta_R)$
\State {\emph MVBPD$(G,s,\eta)$} \label{line: PB end}
\Statex {\bf Procedure} {\emph MVBPD$(G,s,\eta)$} 
\For{$i=1,\ldots,n$} \label{line: neighbourhood begin}
    \State $G_i\leftarrow G[\{v_i\}\cup N_{>\eta}(v_i)\cup  N^2_{>\eta}(v_i)\cup N^3_{>\eta}(v_i)]$ \label{line:neighbourhood decomposition}
        \For{each $S\subseteq N^3_{>\eta}(v_i) $ with $|S|\leq s$} \label{line: enumeration decomposition}
            \State run $updateBounds$ to update $ub_L,ub_R,lb_L,lb_R$
            \State apply reduction rules on $G_i$. \label{line: reduction_fastmvbp}
            \State $MVBPF(G_i,s,\{v_i\}\cup S,N_{>\eta}(v_i)\cup  N^2_{>\eta}(v_i),\newline \{v_1,\ldots,v_{i-1}\}\cup N^3_{>\eta}(v_i)\setminus S)$  \label{line: decomposition end}
        \EndFor
    \EndFor

\end{algorithmic}
\end{algorithm}

\begin{figure}
    \centering
    \includegraphics[width=\linewidth]{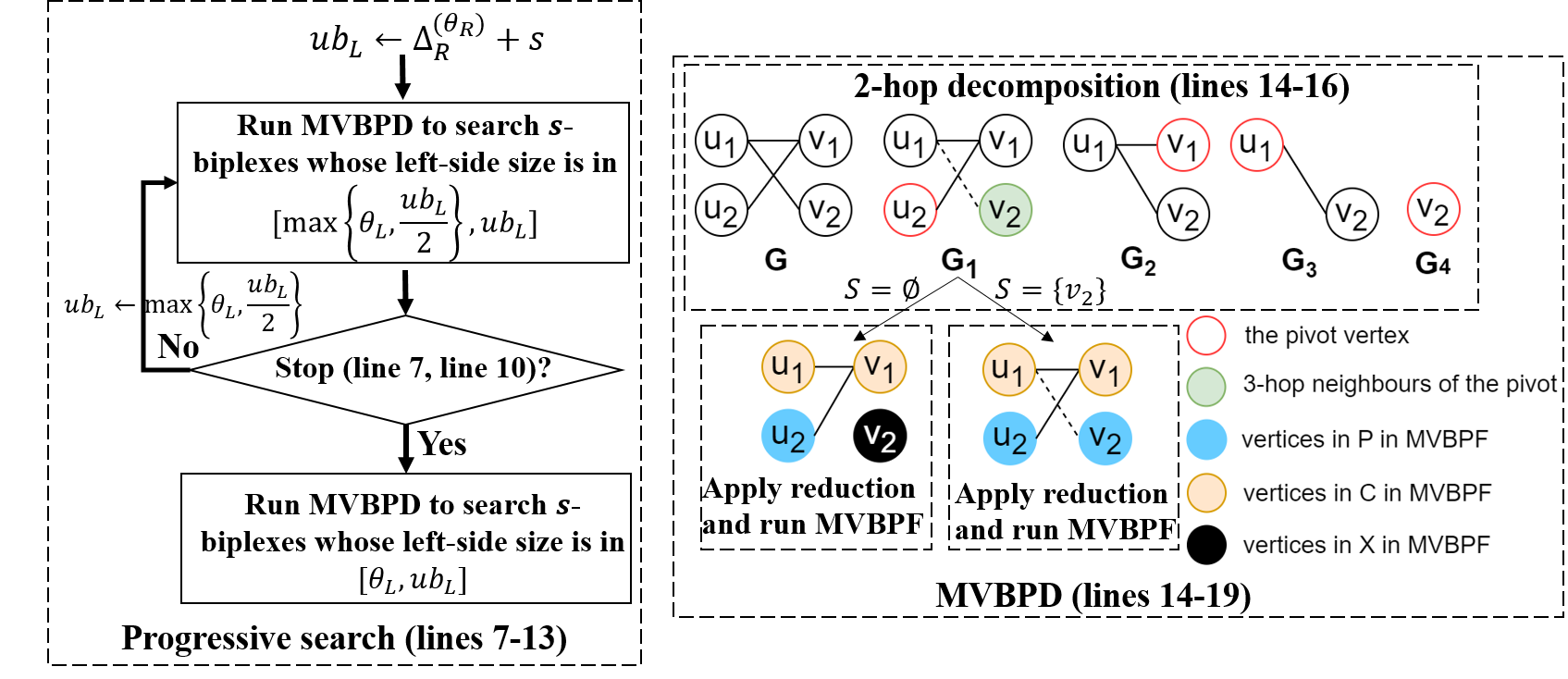}
    \caption{An illustration of key components in the FastMVBP algorithm.}
    \label{fig:fastmvbp framework}
\end{figure}

\subsection{2-Hop Decomposition} \label{sec: 2-hop decomposition}
Towards minimizing the running time, the graph decomposition technique has been widely applied in enumerating and finding subgraph structures \citep{conte2018d2k,zhou2020enumerating,wang2022listing}.
Essentially, it partitions the global graph to multiple local neighbourhoods which allows overlaps. Then the subgraph search algorithm is run on each neighbourhood. 
The final output is the best solutions among all runs. 

Traditionally, the 2-hop neighbourhoods are used to find cliques and bicliques \citep{chen2021efficient, yang2021p, bonchi2019distance}, due to the fact that the diameter of a clique or biclique is upper-bounded by $2$.
For finding biplexes, however, the 2-hop neighbourhoods are not sufficient because the diameter of a biplex might be larger than $2$. Instead, 
\citep{yu2023maximum} introduced a technique to decompose the graph into 3-hop neighbourhoods of vertices whose sizes (number of vertices) are bounded by $\Delta^3$. However, $\Delta^3$ might be a large value, not reducing the running time.

In the following, we provide a novel method to decompose the graph by 2-hop neighbourhoods of vertices. 
The proposed method, named \textit{2-hop decomposition}, is based on \emph{2-hop degeneracy ordering}.
Particularly, the new method can decompose the graph into $n(d_2\Delta)^s$ subgraphs whose sizes are bounded by $d_2+s+1$. $d_2$ is the \textit{largest 2-hop degeneracy} which is even smaller than $\Delta^2$.

\subsubsection{2-hop degeneracy ordering}


A 2-hop degeneracy ordering of graph $G$ is an order of vertices $\eta=\{v_1,\ldots,v_n\}$ such that
each vertex $v_i$ has the minimum 2-hop degree (i.e.,$|N(v_i)\cup N^2(v_i)|$) in the induced subgraph $G[{v_i,\ldots,v_n}]$.
The following introduces the notation based on 2-hop degeneracy ordering.
Denote $d_2$ as the largest 2-hop degeneracy (i.e., $(N(v_i)\cup N^2(v_i))\cap \{v_{i+1},\ldots,v_n\}$) among $i=1,2,\ldots,n$.
Given two vertices $v_i$ and $v_j$, write $v_i <_\eta v_j$ if $v_i$ precedes $v_j$ in $\eta$. Given a vertex $v_i$, define $N_{>_\eta}(v_i):=N(v_i)\cap \{v_{i+1},\ldots,v_n\}$ and $N_{<_\eta}(v_i):=N(v_i)\cap \{v_1,\ldots,v_{i-1}\}$. Also, $N^k_{>_\eta}:=N^k(v_i)\cap \{v_{i+1},\ldots,v_n\}$ and $N^k_{<_\eta}:=N^k(v_i)\cap \{v_1,\ldots,v_{i-1}\}$.

The computation of 2-hop degeneracy can be implemented by iteratively identifying and removing the vertex $v_i$ that has the minimum 2-hop degree in the residual subgraph. Typically, the order in which these vertices are removed corresponds to the 2-hop degeneracy ordering. This process is similar to the computation of 1-hop degeneracy ordering \citep{matula1983smallest}.
However, we point out that the computation of 2-hop degeneracy is more complex. Specifically, after removing a vertex $v_i$, the 1-hop degeneracy computation only needs to decrease the degree of its 1-hop neighbours by one, with a computational complexity of $O(\Delta)$ for each removal. 
In contrast, the update of 2-hop degree requires not only decreasing the 2-hop degree of vertices in $N^2(v_i) \cup N(v_i)$ by one, but also checking whether a vertex in $N(v_i)$ is no longer a 2-hop neighbour of another vertex in $N(v_i)$. This is because the removal of $v_i$ may increase the distance between its neighbours. 
Consequently, the worst-case computational complexity becomes $O(\Delta^2)$ for each removal. The pseudo-code is given in Algorithm~\ref{alg:computeDegeneracy}.

\begin{algorithm}[ht!]
\caption{computeDegeneracy}
\label{alg:computeDegeneracy}
\begin{algorithmic}[1]
\item[INPUT:] $G(L\cup R,E)$, 2-hop degree $TD:L\cup R\rightarrow \mathbb{N}$
\item[OUTPUT:] 2-degeneracy ordering $\eta$
\State initialise $\eta=\{\}$
\For{$i=1,\ldots,n$}
 \State $v_i\leftarrow arg \min_{v\in L\cup R} TD(v)$
 \For{each $u\in N(v_i)\cup N^2(v_i)$}  $TD(u)\leftarrow TD(u)-1$
 \EndFor
 \For{each pair $s,t\in N(v_i)$}
 \If{$|N(s)\cap N(t)|=1$}
\State $TD(s)\leftarrow TD(s)-1$, $TD(t)\leftarrow TD(t)-1$
 \EndIf
 \EndFor
 \State append $v_i$ to the end of $\eta$ and remove it from $G$
\EndFor
\end{algorithmic}
\end{algorithm}



\subsubsection{Decomposition} 
Our decomposition is based on the degeneracy ordering and the fact that the diameter of an $s$-biplex is not larger than 3.
We first decompose the graph $G$ into $n$ local neighbourhoods $G_i$. 
Each neighbourhood is further decomposed into $(d_2\Delta)^s$ subgraphs.
The following introduces the decomposition one by one.



\noindent {\bf Basic decomposition.}
Given a 2-hop degeneracy ordering $\eta$,  the graph $G$ is decomposed into $n$ subgraphs where each subgraph $G_i$ is a local neighbourhood of a vertex $v_i$, i.e., $\{v_i\}\cup {N_>}_\eta(v_i)\cup {N^2_>}_\eta(v_i) \cup {N^3_>}_\eta(v_i)$. A pseudo-code is shown in line~\ref{line: neighbourhood begin}-\ref{line:neighbourhood decomposition} of Algorithm~\ref{alg:fastmvbp}.

It is clear that $|{N^3_>}_\eta(v_i)|$ is bounded by $\min (\Delta^3, n)$, which could be large.
In fact, any $s$-biplex containing $v_i$ has at most $s$ anti-neighbours of $v_i$ in set ${N^3_>}_\eta(v_i)$. Because $s$ is often a small value, we can treat it as a constant. Hence, it is computational affordable to enumerate all potentials of $v_i$'s anti-neighbours. 

\noindent {\bf 2-hop decomposition.} We further decompose $G_i$ into smaller subgraphs by enumerating all potentials of $v_i$'s anti-neighbours $S$ such that $S\subseteq N^3_{>\eta}(v_i)$ and $|S|\leq s$. 
Each decomposed subgraph, i.e., $\{v_i\}\cup {N_>}_\eta(v_i)\cup {N^2_>}_\eta(v_i) \cup S$, focus on $s$-biplexes that include $S\cup \{v_i\}$. 
A pseudo-code is shown in Line~\ref{line: enumeration decomposition}-\ref{line: decomposition end} of Algorithm~\ref{alg:fastmvbp}.

In the following, we discuss the number and sizes of these decomposed subgraphs.
The size of each decomposed subgraph is $1+s+|{N_>}_\eta(v_i)\cup {N^2_>}_\eta(v_i)|$ which is bounded by $d_2+s+1$.
The number of decomposed subgraphs in $G_i$ is related to the combinatorial number of selecting at most $s$ elements from $|{N^3_>}_\eta(v_i)|$ elements. Note that $|{N^3_>}_\eta(v_i)|+1\leq d_2\Delta$. Thus, the total number of enumerated subgraphs is bounded by $(d_2\Delta)^s$.
Consequently, the branching algorithm becomes more efficient, as it operates within smaller subgraphs albeit a small cost of repetitions.
An example of the decomposition algorithm is shown in Fig.~\ref{fig:decomposition}.

\begin{figure}[ht!]
\centering 
\includegraphics[width=0.8\linewidth]{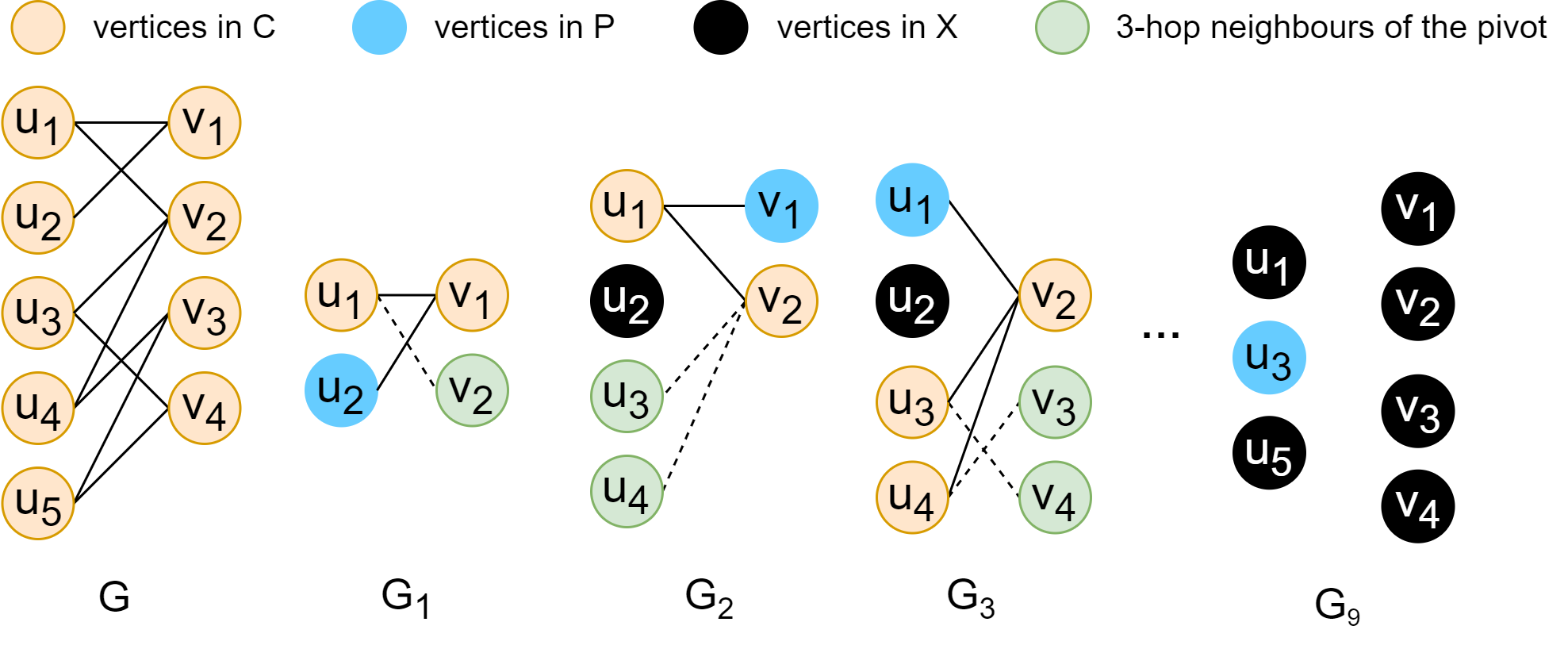}
\caption{An example of 2-hop decomposition. The yellow, blue, and black vertices represent the vertices in $C,P$ and $X$,  respectively. The green vertices represent the 3-hop neighbours of the pivot (the blue vertex).
We set $s=1$ to find $1$-biplexes. 
In graph $G$, $u_2$ has the minimum 2-hop neighbourhood $|N(u_2)\cup N^2(u_2)|=2$ and $N^3(u_2)=\{v_2\}$.
By enumerating $S=\{v_2\}$ and $S=\varnothing$, we get $G_1=\{u_2\}\cup N(u_2)\cup N^2(u_2)\cup S$. 
Next, we move $u_2$ from $C$ to $X$ and then $v_1$ becomes the vertex with the minimum 2-hop neighbourhood in $G[C]$. 
By enumerating $S=\{\{u_3\},\{u_4\},\varnothing \}$
We get $G_2$ by the same operations above. 
Then, we move $v_1$ from $C$ to $X$. 
We repeat the operations above and get $G_1,G_2,G_3,\cdots, G_9$ until $C$ is empty. Finally, the 2-hop degeneracy order is $\eta=\{u_2,v_1,u_1,u_4,v_3,u_5,v_2,v_4,u_3\}$.}
\label{fig:decomposition}
\end{figure}

Note that to exclude the $s$-biplexes that are not maximal, it is necessary to check if there exists a vertex in $v_1,\ldots,v_{i-1}$ can be added to the found $s$-biplex to form a larger $s$-biplex. If such a vertex exists, then the found $s$-biplex is not maximal, and thus the found $s$-biplex will be excluded. We thus add $v_1,\ldots,v_{i-1}$ into the exclusion set $X$, as shown in Line~\ref{line: decomposition end}.
The maximality will be checked according to PR3, as shown in Line~\ref{line:maximal exclusion} of MVBPF.

\subsection{Reduction with Single-side Bounds} \label{sec:reductions with single side bounds}

Kernelization is a crucial method in solving NP-hard problems for large instances \citep{cygan2015parameterized}. The goal of kernelization is to reduce the candidate size without loss of optimality by employing problem-specific {\em reduction rules}.
Related works have validated the effectiveness of reduction rules in finding cohesive subgraphs  \citep{lyu2022maximum,yao2022identifying}. However, most reduction rules cannot be directly applied in TBS problem. Moreover, existing reduction rules often employ a single bound for vertices in two sides \citep{lyu2020maximum,dai2024efficient}.
In this section, we introduce a single-side bound for each side of vertices.

Given a partial solution $P=P_L\cup P_R$, recall that our branching algorithm searches solutions $B$ that $B\supseteq P$ from the candidate set $C=C_L\cup C_R$.
In the following, we introduce the ideas to reduce the size of $C$.
Given a partial solution $P$, for all $s$-biplex induced by $B\supseteq P$,
we consider the bounds on the sizes of $B_L=B\cap L$ and of $B_R=B\cap R$.


\subsubsection{Single-side Bounds}
We focus on four bounds: $ub_L, ub_R, lb_L$ and $lb_R$. 
Specifically, $ub_L$ (resp. $ub_R$) denotes the upper bound on the size of $L\cap B$ (resp. $R\cap B$), and $lb_L$ (resp. $lb_R$) denotes the lower bound on the size of $L\cap B$ (resp. $R\cap B$). 
We call these bounds single-sided bounds as they bound the number of vertices in either $B_L$ or $B_R$.

To determine the value of these bounds, we provide the following three lemmas.



\begin{lemma}\label{lem:observation1}
Given a partial solution $P=P_L\cup P_R$, each $s$-biplex induced by $B\supseteq P$ satisfies
$|B_L|\leq \min_{u\in P_R} |N_{G}(u)| + s$ and $|B_R|\leq \min_{u\in P_L} |N_{G}(u)| + s$.
\end{lemma}
 
\begin{proof}
For each vertex $u\in P_L$, we have $|\overline{N}_{B_R}(u)|\leq s$. And thus $|B_R|=|N_{B_R}(u)|+|\overline{N}_{B_R}(u)|\leq |N_{G}(u)|+s$. Because that this inequation holds for each vertex $u\in P_L$, we conclude that  ${|B_R|}\leq \min_{u\in P_L} |N_{G}(u)| + s$. The same arguments hold for ${|B_L|}$.
\end{proof}

We define that $P_L$ is \textit{fixed} if $|P_L|=ub_L$ or $C_L=\varnothing$, and $P_R$ is \textit{fixed} if $|P_R|=ub_R$ or $C_R=\varnothing$.
In practice, each branch of Algorithm~\ref{alg:mvbp_topk} ends up with 
both $P_L$ is \textit{fixed} and $P_R$ is  \textit{fixed}.
This motivates us to introduce upper bounds when $P_L$ or $P_R$ is fixed.

\begin{lemma} \label{lem:observation2}
Given a partial solution $P=P_L\cup P_R$ and $P_R$ (resp. $P_L$) is \textbf{fixed}, each $s$-biplex induced by $B\supseteq P$ satisfies $|B_L|\leq |P_L|+|\bigcap_{v\in P_R} N_{C_L}(v)|+\sum_{u_i\in P_R} \min\{s-|\overline{N}_{P_L}(u_i)|,|\overline{N}_{C_L}(u_i)|\}$ (resp. $|B_R|\leq |P_R|+|\bigcap_{v\in P_L} N_{C_R}(v)|+
\sum_{u_i\in L} \min\{s-|\overline{N}_{P_R}(u_i)|,|\overline{N}_{C_R}(u_i)|\}$).
\end{lemma}

\begin{proof}
Given that $P_R$ is fixed. We consider the maximum number of vertices that can be included in $P_L$. 
Firstly, the common neighbours $\bigcap_{v\in P_R} N_{C_L}(v)$ can be safely included in $P_L$ without deteriorating the $s$-biplex property.
Secondly, each $u_i\in P_R$ has at most $\min\{s-|\overline{N}_{P_L}(u_i)|, |{ \overline{N}_{C_L}(u_i)}|\}$ anti-neighbours in $C_L$ that can be further included by $P_L$.
In summary, the size of $B_L$ is bounded by $|P_L|+|\bigcap_{v\in P_R} N_{C_L}(v)|+\sum_{u_i\in P_R} \min\{s-|\overline{N}_{P_L}(u_i)|,|\overline{N}_{C_L}(u_i)|\}$ from above.
The same arguments hold for $B_R$ when $P_L$ is fixed.
\end{proof}

\begin{lemma} \label{lem:observation3}
Given a partial solution $P$ and a positive integer $lb$, each $s$-biplex induced by $B\supseteq P$ such that $|B|>lb$ satisfies $|B_L|\geq lb+1-ub_R$, and $|B_ R|\geq lb+1-ub_L$.


\end{lemma}
\begin{proof}
{  Since $|B_L|+|B_R|=|B|\geq lb+1$ and $|B_R|\leq ub_R$, we have $|B_L|\geq lb+1-|B_R|\geq lb+1-ub_R$.} Consequently, any $s$-biplex of size at least $lb+1$ has at least $lb+1-ub_R$ vertices in the left side, meaning $lb_L\geq lb+1-ub_R$. The symmetric arguments hold for $lb_R$.
\end{proof}


During the running of MVBP, we iteratively update the single-side bounds, as described in Algorithm~\ref{alg:updatebounds}. In particular, we update $ub_L$ and $ub_R$ with Lemma~\ref{lem:observation1} and Lemma~\ref{lem:observation2}.
Suppose that $\mathcal{S}$ is the largest $k$ $s$-biplexes found so far, and that the size of the smallest set in the  $\mathcal{S}$ is $lb$, we update $lb_L$ and $lb_R$ with Lemma~\ref{lem:observation3}.





\begin{algorithm}
\caption{updateBounds}
\label{alg:updatebounds}
\begin{algorithmic}[1]
\item[INPUT:] $G(L\cup R,E)$,$lb$,$ub_L$,$ub_R$,$lb_L$,$lb_R$
\item[OUTPUT:] updated $ub_L$,$ub_R$,$lb_L$,$lb_R$
 \State find the vertex $v$ with the minimum $|N_R(v)|$ in $L$, and the vertex $u$ with the minimum $|N_L(v)|$ in $R$
    \State $ub_R\leftarrow \min(Update(G,P_L,P_R,C_L,C_R),N_R(v)+s)$
    \State $ub_L\leftarrow \min(Update(G,P_R,P_L,C_R,C_L),N_L(u)+s)$
    \State $lb_L\leftarrow \max(lb_L,lb+1-ub_R)$
    \State $lb_R\leftarrow \max(lb_R,lb+1-ub_L)$
\Statex {\bf Procedure} {\emph Update$(G,P_{\mu},P_{\nu},C_{\mu},C_{\nu})$} 
    \If{$C_{\mu}=\emptyset$ or $P_{\mu}=ub_{\mu}$}
    \State $sum \leftarrow \sum_{u_i\in P_{\mu}\cup C_{\mu}} \min(s-|\overline{N}_{P_{\nu}}(u_i)|,|\overline{N}_{C_{\nu}})(u_i)|$
    \State \bf {return} $\min(ub_{\nu},|P_{\nu}|+|\bigcap_{v\in P_{\mu}} N_{C_{\nu}}(v)|+sum)$
    \Else   \bf{ return} $P_{\nu}\cup C_{\nu}$
    \EndIf
\end{algorithmic}
\end{algorithm}

\subsubsection{Reduction Rules with Single-side Bounds.} In this section, we introduce three reduction rules highlighting the single-side bounds.
Given a bipartite graph $G=(L\cup R,E)$, partial solution $P=P_L\cup P_R$, candidate set $C=C_L\cup C_R$, exclusion set $X$, and the single-side bounds $lb_L,lb_R,ub_L,ub_R$.


\noindent 
\textbf{Reduction 1.} 
\label{redu:1}
If there exists a vertex $u \in C_L$ (resp. $u \in C_R$) such that $|N(u)|+s<lb_R$ (resp.  $|N(u)|+s<lb_L$), then $u$ is not included in any $s$-biplexes $B$ with $|B_R|\geq lb_R$ (resp. $|B_L|\geq lb_L$).

This is because that $lb_R\leq |B_R|\leq |N(v)|+s$ for each vertex $v\in B_L$. Similarly, $lb_L\leq |B_L|\leq |N(v)|+s$ for each vertex $v\in B_R$.

\noindent 
\textbf{Reduction 2.} If there exists a vertex $u\in P_L$ (resp. $v \in P_R$) such that $|N(u)|=lb_R-s$ (resp. $|N(v)|=lb_L-s$), then $N(u)$ (resp. $N(v)$) must be included in each $s$-biplex that contains $P$. Moreover, given that $P\supseteq N(u)$, 
each vertex $x\in C_L$ (resp. $x \in C_R$) such that 
$|\overline{N}_{P_R}(x)|>s$ is not included in any $s$-biplex that contains $P$.

The proof of Reduction 2 is two-fold.
Firstly, the result that $N(u)$ must be included directly follows from the property of $s$-biplex and the definition of $lb_R$.
Secondly, the vertex $x$ is not included in any $s$-biplex as $x$ has more than $s$ anti-neighbours in partial solution $P$.

The following two reductions directly follow from the definition of single-side bounds.

\noindent 
\textbf{Reduction 3.} If $|P_R|+|C_R|=lb_R$ (resp. $|P_L|+|C_L|=lb_L$), then $P_R\cup C_R$ (resp. $P_L\cup C_L$) must be included by any $s$-biplex that contains $P$.

\noindent \textbf{Reduction 4.} If either $lb_L>ub_L$ or $lb_R>ub_R$, then there does not exist an $s$-biplex that contains $P$ and meets the constraints.

The following reduction follows from the Section 4.2 of \citep{yu2023maximum}.

\noindent \textbf{Reduction 5.}
\label{redu:5}
Each pair of vertices $u,v\in L$ (resp. $u,v\in R$ ) with $|N(u)\cap N(v)|<lb_R-2s$ (resp. $|N(u)\cap N(v)|<lb_L-2s$) can not be simultaneously included in any $s$-biplex that meets the constraints.

The reduction rules above are applied in line~\ref{line: reduction_fastmvbp} in Algorithm~\ref{alg:fastmvbp}.

\subsection{Progressive Search} \label{sec: progressive search}
In top-$k$ cohesive subgraph search problems, an algorithm often find many small subgraphs which do not appear in the final solutions.
To defer the search of these small subgraphs, the progressive search technique is proposed \citep{lyu2020maximum}.
The main idea is to prioritize the search of large cohesive subgraphs by heuristically adjusting the constraints, including both size lower bounds and size upper bounds.
The aforementioned reduction techniques can be combined with the constraints to further reduce the graph size, resulting in less running time of the branching algorithms.

In this section, we introduce the progressive search technique to partition the bounds on left-side vertices  $[\theta_L,\Delta^{(\theta_R)}_R+s]$ and the bounds on right-side vertices $[\theta_R,\Delta^{(\theta_L)}_L+s]$ into disjoint sub-ranges. 

Initialise five variables $lb\leftarrow 0$,
$lb_L\leftarrow \theta_L$, $ub_L\leftarrow \Delta^{(\theta_R)}_R+s$, $lb_R\leftarrow \theta_R$ and $ub_R\leftarrow \Delta^{(\theta_L)}_L+s$.
We partition the search range into multiple sub-ranges by iteratively updating the variables:


(i) Run {\emph MVBPD} to search $s$-biplexes whose number of left-side vertices fits in $[\Phi,ub_L]$ where $\Phi=\max (ub_L/2,\theta_L)$, and number of right-side vertices fits in $[\theta_R,\Delta^{(\theta_L)}_L+s]$.
After this run, if we have found $k$ maximal $s$-biplexes, then update $lb\leftarrow$ the size of the smallest one among these biplexes, $ub_L\leftarrow { lb_L}$, and go to (iii).
Otherwise, go to (ii).


(ii) If $lb_L> \theta_L$, then set $ub_L\leftarrow \Phi, \Phi\leftarrow max(\frac{\Phi}{2},\theta_L), lb_L\leftarrow \Phi$ and go to (i). Otherwise, go to (iii). 

(iii) Run {\emph MVBPD} to search $s$-biplexes whose number of left-side vertices fit in $[\theta_L,ub_L]$, and number of right-side vertices fit in $[\max(lb+1-ub_L,\theta_R),ub_R]$. After this run, we terminate.

Notice that the union of the disjoint sub-ranges for left-side vertices is exactly $[\theta_L,\Delta^{(\theta_R)}_R+s]$. And we search each disjoint sub-range with fixed right-side range $[\theta_R,\Delta^{(\theta_L)}_L+s]$. These two setups  guarantee to cover all possible maximal $s$-biplexes.
We employ an early stop procedure once $k$ $s$-biplexes have been found.
In this case, we move to search larger $s$-biplexes whose number of left-side vertices fits in $[\theta_L, ub_L]$. In addition, by applying Lemma~3, we improve the range for right-side vertices from $[\theta_R,\Delta^{(\theta_L)}_L+s]$ to $[\max(lb+1-ub_L,\theta_R),\Delta^{(\theta_L)}_L+s]$. The Progressive Search technique is employed in line~\ref{line: PB}-\ref{line: PB end} in Algorithm~\ref{alg:fastmvbp}.

\subsection{Time Complexity of FastMVBP}
In this section, we analyse the worst-case time complexity of the improved algorithm FastMVBP. 
\begin{theorem}
Given a bipartite graph $G=(L\cup R,E)$,
FastMVBP finds top-$k$ maximal $s$-biplexes in
$O(n\Delta^{s+1}d_2^{s+1} {\gamma_s}^{d_2})$ time, where $\gamma_s$ is the largest real root of $x^{s+4} - 2x^{s+3} + x^2 -x +1 =0$, which is strictly less than 2.
{ For $s=0,1,2,3$, $\gamma_s=1.465,1.754,1.888, 1.947$, respectively.}
\end{theorem}
According to Section$~\ref{BaseTimeComplexity}$, the time complexity of MVBP is $O(\gamma_s^{|C|})$, where $|C|\leq n$. Next, we introduce how FastMVBP reduce $|C|$ to $d_2$. 
The FastMVBP algorithm computes the 2-hop degeneracy order once in $O(n\Delta^2)$ time.  
Given the 2-hop degeneracy order and a vertex $v_i$, we enumerate all subsets $S\subseteq N^3_{>\eta}(v_i)$ with $|S|\leq s$. In particular, the number of enumerated subsets is bounded by $|N^3_{>\eta}(v_i)|^s\leq (d_2\Delta)^s$.
In each decomposed subgraph $\{v_i\}\cup {N_>}_\eta(v_i)\cup {N^2_>}_\eta(v_i) \cup S$, we initialise with $|C|=|{N_>}_\eta(v_i)\cup {N^2_>}_\eta(v_i)|\leq d_2$, and thus FastMVBP runs in $O({\gamma_s}^{d_2})$. 
Additionally, the time for checking the maximality of found $s$-biplexes is $O(d_2+d_2\Delta)$. 
The total running time is
$O(n\Delta^2+ (d_2+d_2\Delta)\sum_{v_i\in V}|N^3_{>\eta}(v_i)|^s{\gamma_s}^{d_2})$ which is bounded by
$O(n\Delta^2 + n\Delta^{s+1}d_2^{s+1}{\gamma_s}^{d_2})$. This result is particularly useful as the parameter $s$ is often small and fixed in real-world tasks.
\section{Experiments}
\subsection{Experiment Setup}
\par
\begin{table}[h]
\centering
  \caption{Dataset statistics}
  \label{tab:tbl-data}
  \resizebox{1.0\columnwidth}{!}{
      \begin{tabular}{c|c|c|c|c|c}
        \toprule[2pt]
        Dataset & $|L|$ & $|R|$ & $|E|$ & $\Delta_L$ &$\Delta_R$ \\
        \hline
        YouTube & 94,238 & 30,087 & 293,360 & 1,035 & 7,591 \\
        Location & 172,091 & 53,407 & 293,697 & 28 & 12,189 \\
        CiteULike & 153,277 & 731,769 & 2,338,554 & 189,292& 1264\\
        Twitter & 175,214 & 530,418 & 1,890,661 & 968 & 19,805\\
        IMDB & 896,302 & 303,617 & 3,782,463 & 1,334 & 1,590 \\
        Team & 901,166 & 34,461 & 1,366,466 & 17 &2,671  \\
        AmazonRatings & 2,146,057 & 1,230,915 & 5,743,258 & 12,180 & 3,096 \\
        Trackers & 27,665,730 & 12,756,244 & 140,613,762 & 1,100,065 & 11,571,952 \\
		\bottomrule[2pt]
\end{tabular}
}
\end{table}

\noindent {\bf Datasets} We evaluate our algorithm on eight real-world graphs collected from KONECT (\url{http://konect.cc/}).
The number of edges of these datasets ranges from $293K$ to $140M$. The maximum degree of vertices of these datasets ranges from $17$ to $11M$.  
The statistics of these datasets are shown in Table~\ref{tab:tbl-data}.
These datasets cover a wide range of real-world scenarios including social networks, geometric location networks, citation networks, web networks and cooperation networks.  



\smallskip
\noindent {\bf Benchmarks.}
We compare FastMVBP with two state-of-the-art algorithms.
\begin{itemize}[leftmargin=*]
    \item \textit{FastBB.} The FastBB is designed for searching $s$-biplexes with top-$k$ largest number of edges \citep{yu2023maximum}.
    To solve the TBS problem, we modified the code of FastBB  by removing reductions related to most edges, changing progressively search to focus on most vertices, and keeping track of solutions with most vertices.
    
    \item \textit{ListPlex.}The ListPlex is known to be one of the most efficient algorithms for enumerating all maximal $s$-plexes from a normal graph \citep{wang2022listing}. 
    To list $s$-biplexes in bipartite graphs, we modify ListPlex by considering both left side and right side vertices, removing reductions that could not be applied in bipartite graphs,  and keeping track of $s$-biplexes with most vertices.
\end{itemize}

\smallskip
\noindent {\bf Implementation.}
The code is implemented using C++. All experiments are conducted on a single machine with an Intel(R) Xeon(R) Platinum 8360Y CPU @ 2.40GHz. All source codes are available at \url{https://github.com/PlutoAiyi/FastMVBP}. \par

\smallskip
\noindent {\bf Settings.} 
We conduct seven experiments.
In the first three experiments, we test the running times of different algorithms on real-world datasets and investigate the effects of different parameters.
We also conduct two experiments on synthetic ER models\citep{bollobas1998random}
to investigate the effect of the datasets themselves.
In the last two experiments, we carried out case studies to discuss the effect of different components in FastMVBP and showcase the difference between solutions of TBS and of TBSE.
We set the time limit as 24 hours (86400 seconds). 
If an instance runs more than 24 hours, we mark the instances \textit{OOT} as out of time. 
In our experiments, we count the 2-degeneracy $d_2$ after applying reduction rules. 

{ 
In real-world applications, the parameter $s$ is often set to a small value aiming to find more cohesive and dense subgraphs. Notice that $\theta_L$ and $\theta_R$ indicate the size lowerbounds of left side and right side of $s$-biplexes. When $\theta_L$ and $\theta_R$ are large, few $s$-biplexes that can satisfy the size lowerbounds. Therefore, these two parameters are also set to small values. These parameter settings are widely adopted by $s$-plex search papers such as \citep{yu2023maximum,wang2022listing,dai2024efficient}. 
}


\subsection{Experiment Results}
\begin{table}[ht!]
    \centering
	\caption{The running time of listing top-10 maximal $s$-biplexes in real-world datasets with varying $s$}
	\label{tab:tbl-community1}
    \resizebox{0.8\textwidth}{!}{
    \begin{tabular}{c|c|c|c|c|c}
	\toprule[1.5pt]
	\multirow{2}{*}{\tabincell{c}{Graph\\ $(|L|,|R|,|E|)$}} & \multirow{2}{*}{$s$}    & \multirow{2}{*}{$\theta_L,\theta_R$}  &  \multicolumn{3}{c}{The running time (s)}
	\\
	\cline{4-6}
	& & &  FastBB & ListPlex & \textbf{FastMVBP}\\
	\hline
 	\multirow{3}{*}{\shortstack{YouTube \\$(94238,30087,293360)$}} & 1 & 3,3 & 532.359 &14260.9 & \textbf{10.644}\\
    \cline{2-6}
	& 3 & 15,15 & OOT & 17100.4 & \textbf{196.715}\\
     \cline{2-6}
	& 5 & 20,20 & OOT & 61343.9 &\textbf{2665.68} \\
	\cline{1-6}
  \multirow{2}{*}{\shortstack{Location\\ $(172091,53407,293697)$} } & 1 & 3,3 & 61.555 & 3769.85 &\textbf{2.107}\\
 \cline{2-6}
	& 2 & 5,5 & \textbf{0.117} &5004.44 &0.293\\
 \cline{1-6}
	\multirow{3}{*}{\shortstack{CiteULike\\ $(153277,731769,2338554)$}} & 1 & 11,11 & 2165.52 & 501.524 & \textbf{132.364} \\
     \cline{2-6}
	& 3 & 20,20 & OOT &57.912 &\textbf{17.23} \\
      \cline{2-6}
	& 5 & 24,24 & OOT & 3015.6 &\textbf{153.243} \\
	\cline{1-6}
	\multirow{3}{*}{\shortstack{Twitter\\ $(175214,530418,1890661)$}} & 1 & 9,9 & 240.734 & 93.153& \textbf{6.461}\\
      \cline{2-6}
	& 3 & 16,16 & 21121.6 & 7900.65 &\textbf{5.869}\\
      \cline{2-6}
	& 5 & 21,21 & 19438.6 & 548.69 &\textbf{4.322}\\
	\cline{1-6}
	\multirow{3}{*}{\shortstack{IMDB\\ $(896302,303617,3782463)$}} & 1 & 3,3 & 208.019 & OOT&\textbf{125.608}\\
 \cline{2-6}
	& 3 & 13,13 & 33601.2 &OOT &\textbf{112.613}\\
 \cline{2-6}
	& 5 & 25,25 & 3450.89 & 104.08 &\textbf{19.090}\\
	\cline{1-6}
 	\multirow{2}{*}{\shortstack{Team\\ $(901166,34461,1366466)$}} & 1& 3,3 & 4901.18 & 20175.7&\textbf{765.875}\\
	\cline{2-6}
	& 2 & 7,7 & 892.157 &247.579 &\textbf{39.671} \\
 \cline{1-6}
 	\multirow{3}{*}{\shortstack{AmazonRatings\\ $(2146057,1230915,5743258)$}} & 1 & 7,7 & 191.635 &657.838 &\textbf{9.567}\\
 \cline{2-6}
	& 3 & 11,11 & OOT &OOT &\textbf{22.221}\\
 \cline{2-6}
	& 5 & 17,17 & OOT &OOT &\textbf{962.669}\\
    \cline{1-6}
	\multirow{2}{*}{\shortstack{Trackers\\$(27665730,12756244,140613762)$}} & 1 & 150,150 & 1338.76 & 11504.9 &\textbf{27.457}\\
 \cline{2-6}
	& 2 & 150,150 & 7631.21 &OOT &\textbf{24.683}\\
    \bottomrule[1.5pt]
    \end{tabular}
}
\end{table}

\noindent
{\bf Varying $s$.} We vary $s=1,3,5$ for dense bipartite graphs, and $s=1,2$ for sparse or unbalanced bipartite graphs aiming to analyze the effect of $s$. Note that a larger $s$ is meaningless as it may result in sparse subgraphs. We fix $k=10$ and adjust $\theta_L=\theta_R$ to make sure the existence of at least one maximal $s$-biplex.

The results are shown in Table~\ref{tab:tbl-community1}. We observe that FastMVBP is the only algorithm that completes within the time limit across all experimental scenarios. 
In all cases, FastMVBP outperforms ListPlex by at least one order of magnitude. This finding underscores the importance of designing algorithms tailored specifically for bipartite graphs.
In almost all cases, FastMVBP outperforms FastBB by up to three orders of magnitude. The only exception is that FastBB is marginally faster than FastMVBP by 0.2 seconds on the sparse Location dataset with parameters $\theta_L=\theta_R=5$. This is due to the sparse nature of the Location dataset, where the basic branching framework MVBP has already efficiently solved the TBS problem. Further details will be provided in our ablation study.


\begin{table}[ht!]
    \centering
	\caption{The running time of listing top-10 maximal $s$-biplexes from real-world graphs by different algorithms with varying $\theta=\theta_L=\theta_R$ ($s=1$)}
	\label{tab:tbl-community2}
    \resizebox{0.8\textwidth}{!}{
    \begin{tabular}{c|c|c|c|c}
	\toprule[1.5pt]
	\multirow{2}{*}{\tabincell{c}{Graph\\ $(|L|,|R|,|E|)$}} &  \multirow{2}{*}{$\theta_L,\theta_R$}  &  \multicolumn{3}{c}{The running time (s)}
	\\
	\cline{3-5}
	& &  FastBB & ListPlex & \textbf{FastMVBP} \\
	\hline
  	\multirow{3}{*}{\shortstack{YouTube \\$(94238,30087,293360)$}} & \multirow{3}{*}{} 3,3 & 532.359 &1227.69& \textbf{10.644}\\
   &  4,4& 84.391 & 28.7365 & \textbf{1.807}\\
   &  5,5& 97.369 & 44.8598 & \textbf{4.286}\\
	\cline{1-5}
  \multirow{3}{*}{\shortstack{Location\\ $(172091,53407,293697)$} } &  \multirow{3}{*}{}  3,3 & 61.555 & 3769.85 &\textbf{2.107}\\
      &  4,4& 0.259 & 2.130& \textbf{0.162} \\
   &  5,5& 0.014 & \textbf{0.012} & 0.050\\
   \cline{1-5}
	\multirow{3}{*}{\shortstack{CiteULike\\ $(153277,731769,2338554)$}} & \multirow{3}{*}{}  11,11 & 2165.52 &501.524 & \textbf{132.364}\\
   &  12,12& 4162.89 & 302.58 & \textbf{137.948}\\
   &  13,13& 3300.24 & 128.764 & \textbf{75.204}\\
	\cline{1-5}
 	\multirow{3}{*}{\shortstack{Twitter\\ $(175214,530418,1890661)$}} & \multirow{3}{*}{}  7,7 & 13412 &22451.8 & \textbf{146.788}\\
     &  8,8& 73.313 &359.531 & \textbf{16.335}\\
    &  9,9& 240.734 & 93.153 & \textbf{6.461} \\
   	\cline{1-5}
	\multirow{3}{*}{\shortstack{IMDB\\ $(896302,303617,3782463)$}} & \multirow{3}{*}{}  3,3 & 208.019 & OOT&\textbf{125.608}\\
    &  4,4& 82.473 &OOT & \textbf{15.280}\\
   & 5,5& 8.600 & OOT& \textbf{4.926}\\
	\cline{1-5}
 	\multirow{3}{*}{\shortstack{Team\\ $(901166,34461,1366466)$}} & \multirow{3}{*}{} 3,3 & 4901.18 &20175.7 &\textbf{765.875} \\
      &  4,4& 317.176 & 1075.24& \textbf{108.801} \\
   &  5,5& 27.364 & 83.493 & \textbf{19.194}\\
 \cline{1-5}
 	\multirow{3}{*}{\shortstack{AmazonRatings\\ $(2146057,1230915,5743258)$}} & \multirow{3}{*}{}  7,7 & 191.635 & 657.838 &\textbf{9.567}\\
     &  8,8& 88.912 & 199.321 & \textbf{19.271}\\
   &  9,9& 50.689 & 70.139 & \textbf{16.923}\\
    \cline{1-5}
	\multirow{3}{*}{\shortstack{Trackers\\$(27665730,12756244,140613762)$}} & \multirow{3}{*}{}  145,145 & 2789.44 &OOT &\textbf{27.980}\\
      &  150,150& 1338.76 & 11504.9 & \textbf{27.457} \\
   & 155,155& 568.854 & 357.266& \textbf{90.131}\\
    \bottomrule[1.5pt]
    \end{tabular}
}
\end{table}

\noindent
{\bf Varying $\theta_L$ and $\theta_R$.} 
This experiment investigates the effect of $\theta_L$ and $\theta_R$. We fix $k=10$, $s=1$ and vary $\theta_L$ and $\theta_R$ for different graphs. 
Note that small $\theta_L$ and $\theta_R$ may result in extremely unbalanced subgraphs.
Therefore, we adjust $\theta_L$ and $\theta_R$ with large values and make sure the existence of at least one maximal $s$-biplex.
As shown in Table~\ref{tab:tbl-community2}, 
FastMVBP outperforms both benchmarks in almost all cases, which is similar to the results in the {\em varying $s$} experiment.
We observe a clear trend that the running time of all algorithms decreases as $\theta_L$ and $\theta_R$ go large. A possible reason is that large $\theta_L, \theta_R$ additionally add constraints in the TBS search, resulting in more reductions in the graph size. 
This finding highlights the necessity of the progressive search technique, which prioritises the search of large subgraphs. 

\begin{figure*}[ht!]
	\centering
	\begin{minipage}{0.32\linewidth}
		\centering
        \includegraphics[width=\linewidth]{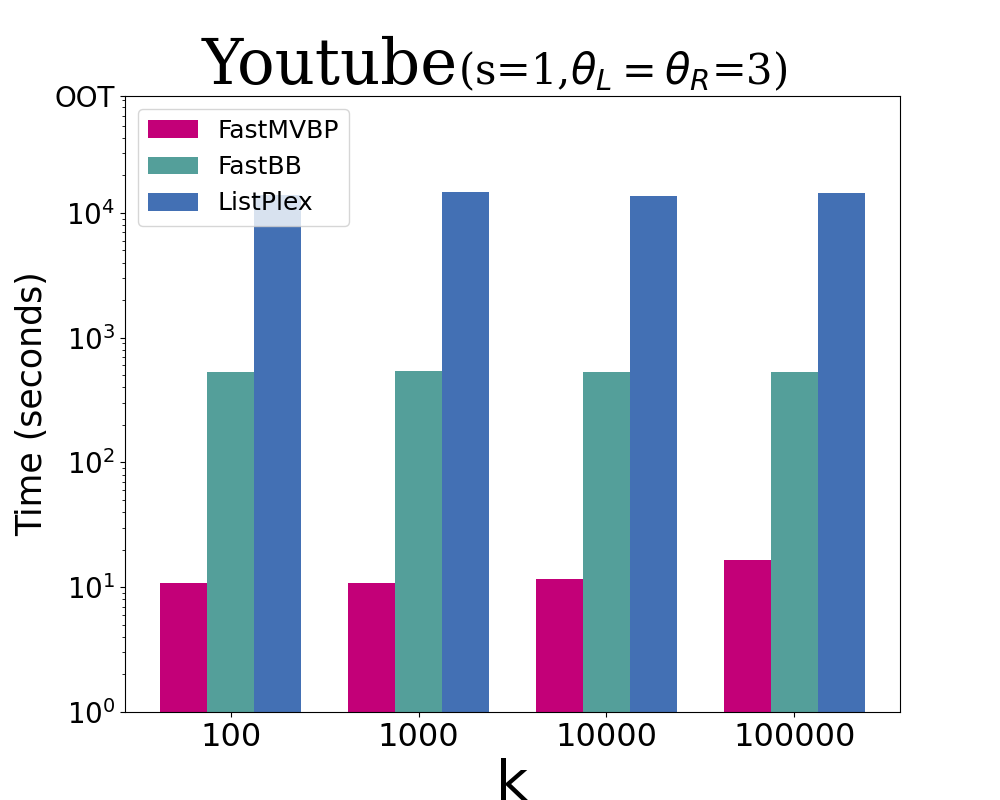}
	\end{minipage}
	\begin{minipage}{0.32\linewidth}
		\centering
		\includegraphics[width=\linewidth]{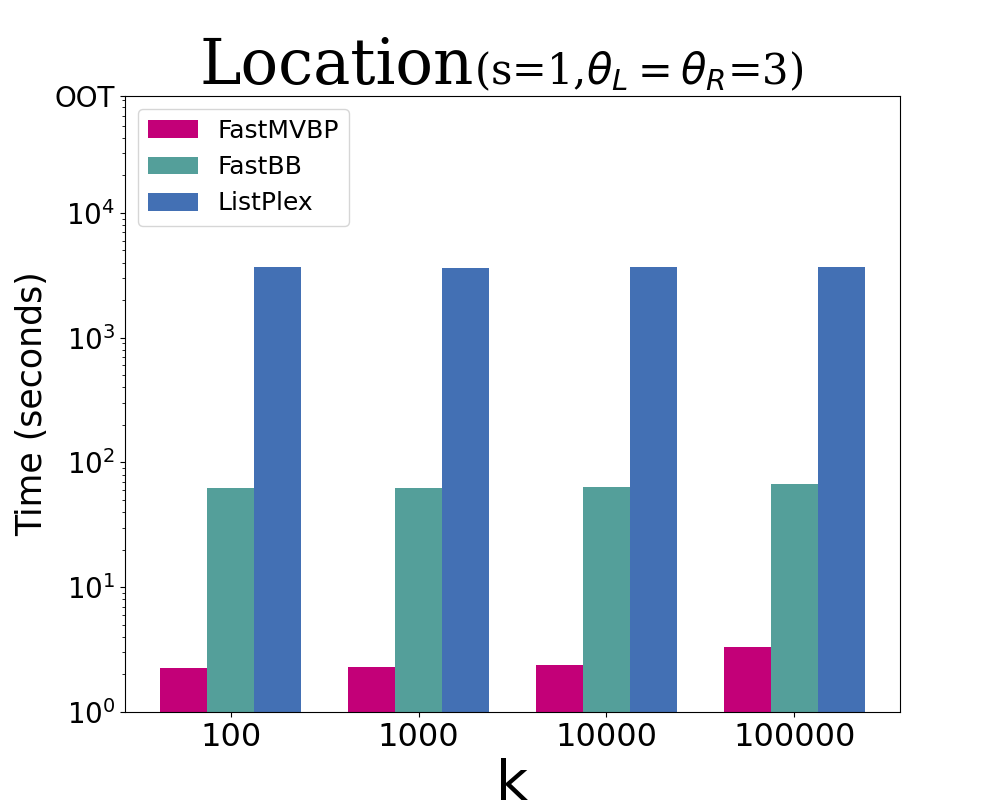}
	\end{minipage}
	\begin{minipage}{0.32\linewidth}
		\centering
		\includegraphics[width=\linewidth]{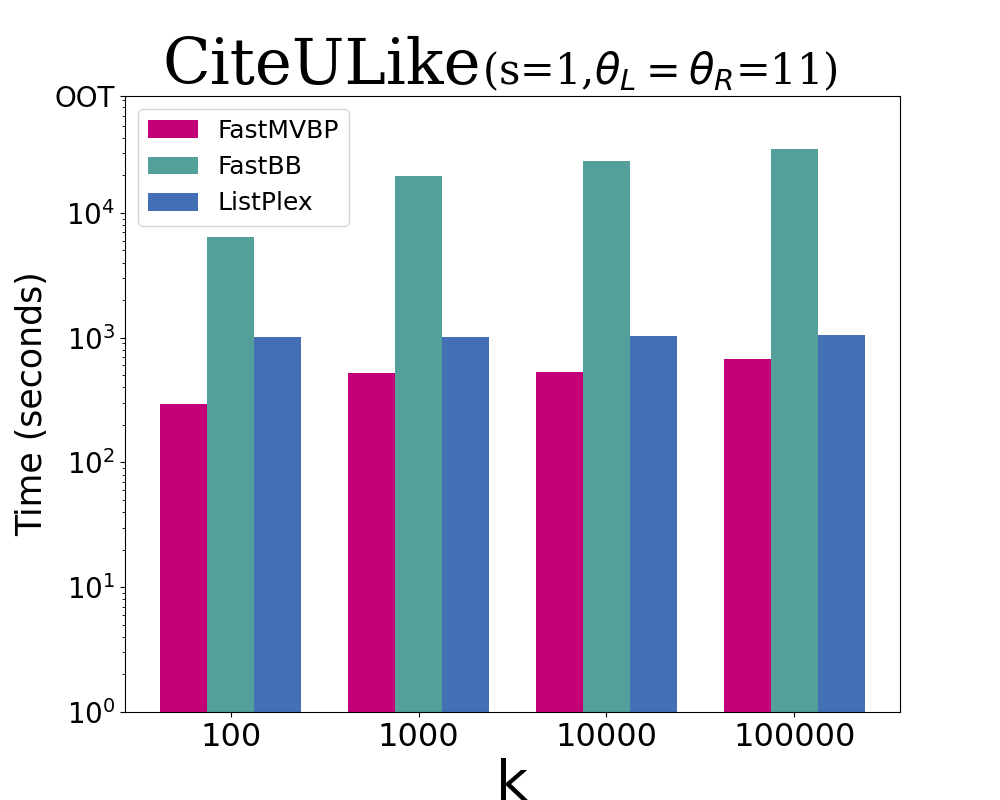}
	\end{minipage}
 
	\begin{minipage}{0.32\linewidth}
		\centering
		\includegraphics[width=\linewidth]{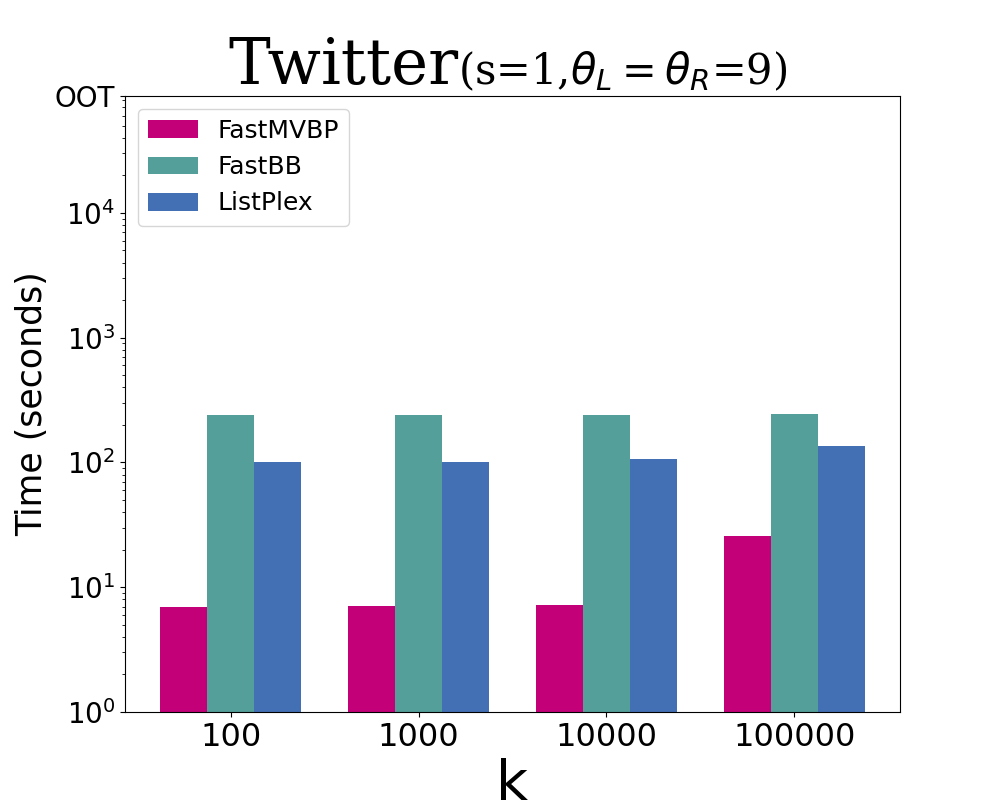}
	\end{minipage}
	\begin{minipage}{0.32\linewidth}
		\centering
		\includegraphics[width=\linewidth]{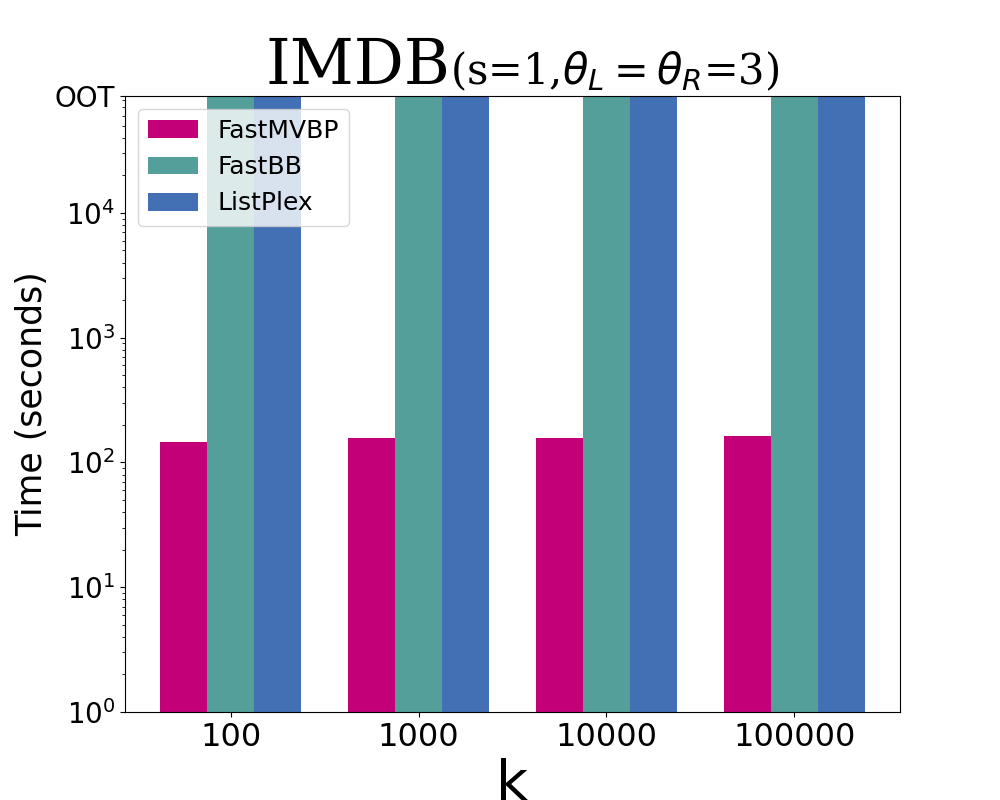}
	\end{minipage}
	\begin{minipage}{0.32\linewidth}
		\centering
		\includegraphics[width=\linewidth]{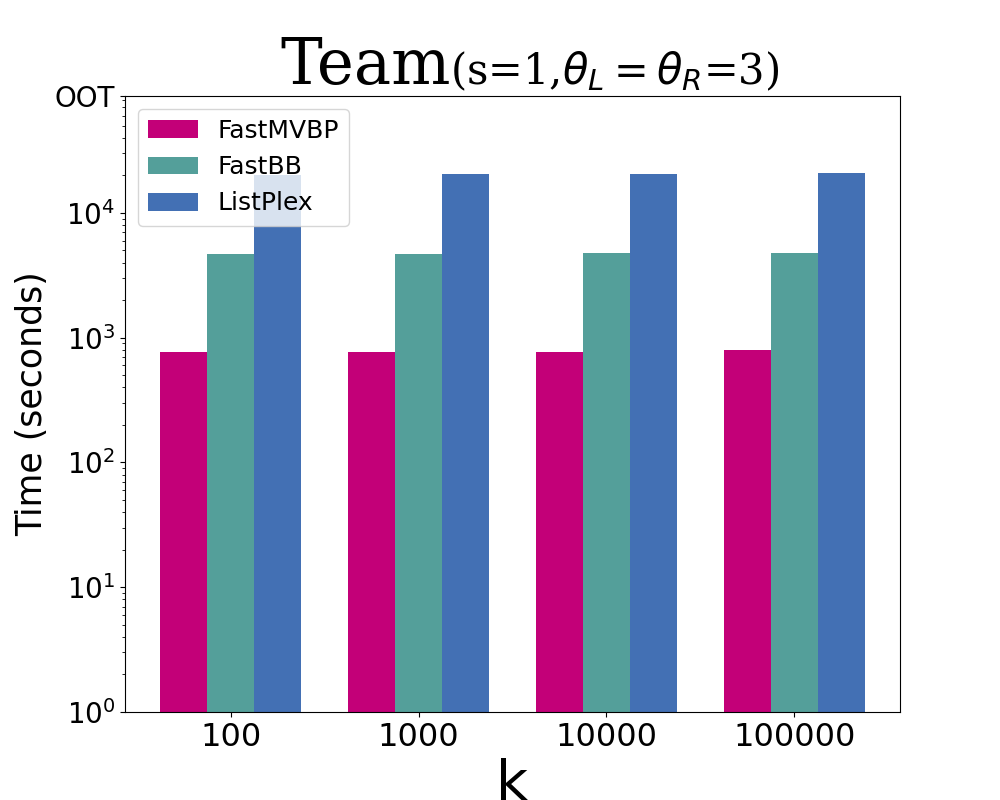}
	\end{minipage}
 
	\begin{minipage}{0.32\linewidth}
		\centering
		\includegraphics[width=\linewidth]{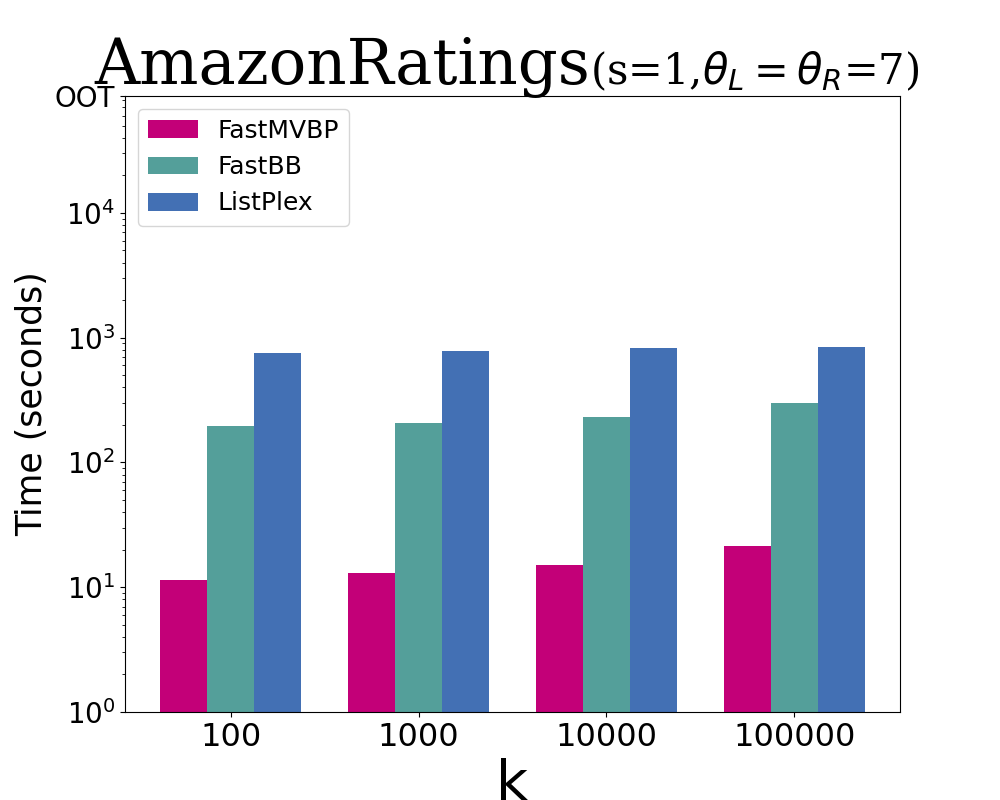}
	\end{minipage}
	\begin{minipage}{0.32\linewidth}
		\centering
		\includegraphics[width=\linewidth]{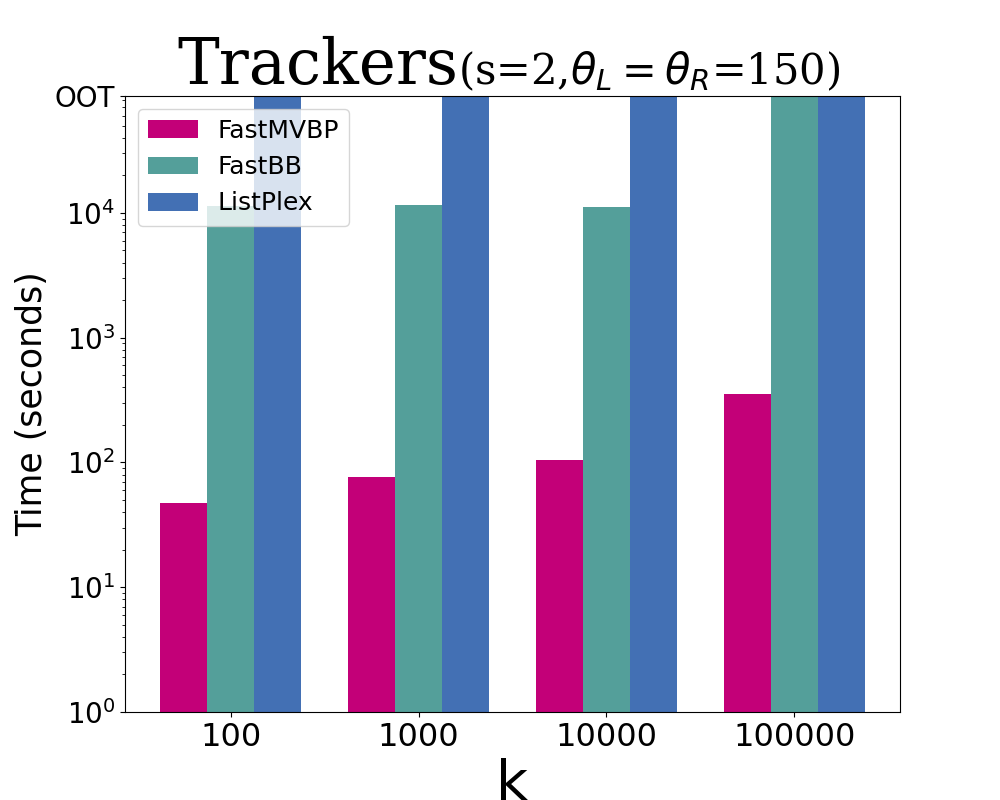}
	\end{minipage}

 \caption{The running time of listing top-$k$ maximal $s$-biplexes with varying $k$}
 \label{fig:varying_k}
\end{figure*}

\noindent
{\bf Varying $k$.}  We fix $s=1, \theta_L=\theta_R$, and vary $k=100,1000,10000,100000$ to investigate the effect of $k$.
As shown in Fig.~\ref{fig:varying_k}, FastMVBP significantly outperforms all benchmarks in all cases.
We observe that both ListPlex and  FastBB failed in all scenarios in the IMDB dataset, while FastMVBP successfully found solutions in several minutes. This finding highlights the efficiency of our proposed FastMVBP algorithm. We also observe significant differences in running times across various datasets. This variation can be attributed to the differing sizes and edge densities of these datasets.

\smallskip
\noindent {\bf Synthetic Datasets.}
To investigate the effect of dataset statistics, we generate synthetic datasets using the well-known ER (Erdös-Réyni) model \citep{bollobas1998random}. 
The density of an ER bipartite graph is defined as $\frac{2m}{|L|+|R|}$.
We control the number of vertices and the density to generate different datasets.

\noindent
{\bf Varying $\# vertices$.}  We conduct experiments on synthetic datasets generated using the ER model.
We fix $k=10, s=1, \theta_L=\theta_R=3$, the dataset density as $20$ and vary the number of vertices from $100$ to $1M$ to investigate the effect of dataset size.
As shown in Fig.~\ref{fig:synthetic} (left), FastMVBP outperforms all benchmarks in all cases.
We observe that FastMVBP consistently outperforms all benchmarks across all scenarios. Furthermore, the running time for all algorithms increases exponentially as the number of vertices grows. This implies that running algorithms on smaller input graphs can significantly reduce running time. These findings underscore the importance of employing both reduction and graph decomposition techniques, which are capable of significantly reducing graph size.

\noindent
{\bf Varying $density$.} We then fix $k=10, s=1, \theta_L=\theta_R=3$, the dataset size as $100,000$ and vary dataset density from $5$ to $40$ to investigate the effect of density in the synthetic ER graphs.
As shown in Fig.~\ref{fig:synthetic} (right), 
FastMVBP consistently outperforms all benchmarks. In addition, we observe that the running time of all algorithms increases as the density goes large. A possible reason is that there are more branches to be considered in dense graphs.
In particular, we observe that FastBB and ListPlex become OOT when density is $40$.
The findings above imply the pruning principles of FastMVBP are more effective.  

\begin{figure}[ht!]
	\centering
	\begin{minipage}{0.4\linewidth}
		\centering
		\includegraphics[width=\linewidth]{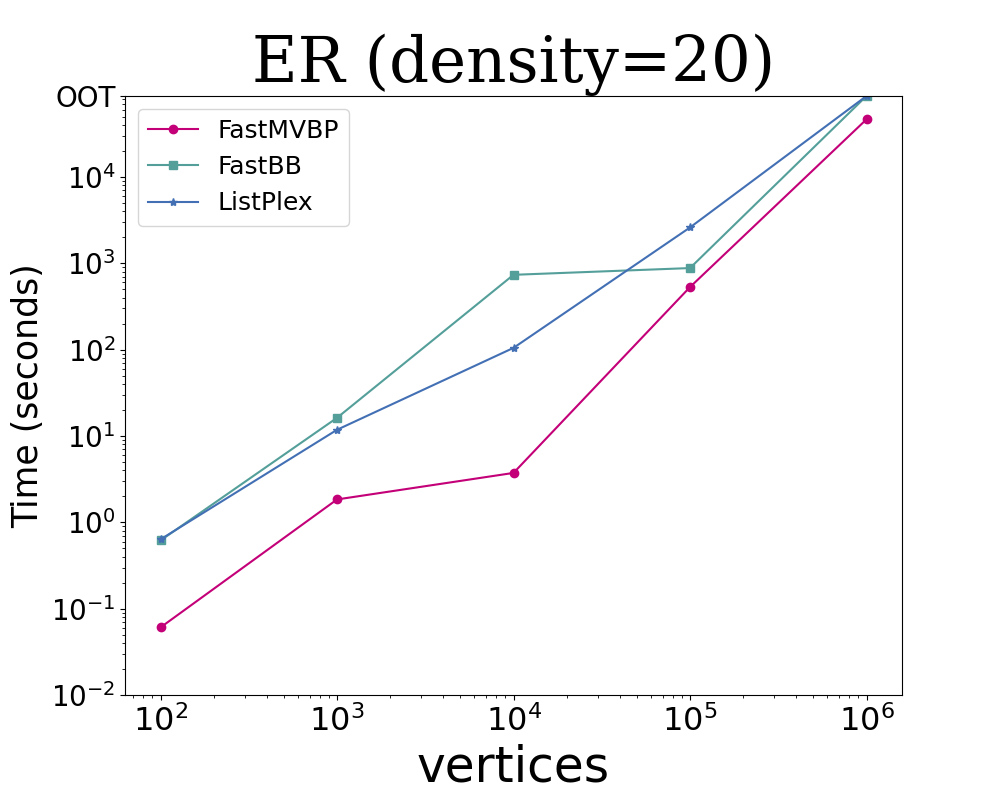}
	\end{minipage}
	\begin{minipage}{0.4\linewidth}
		\centering
		\includegraphics[width=\linewidth]{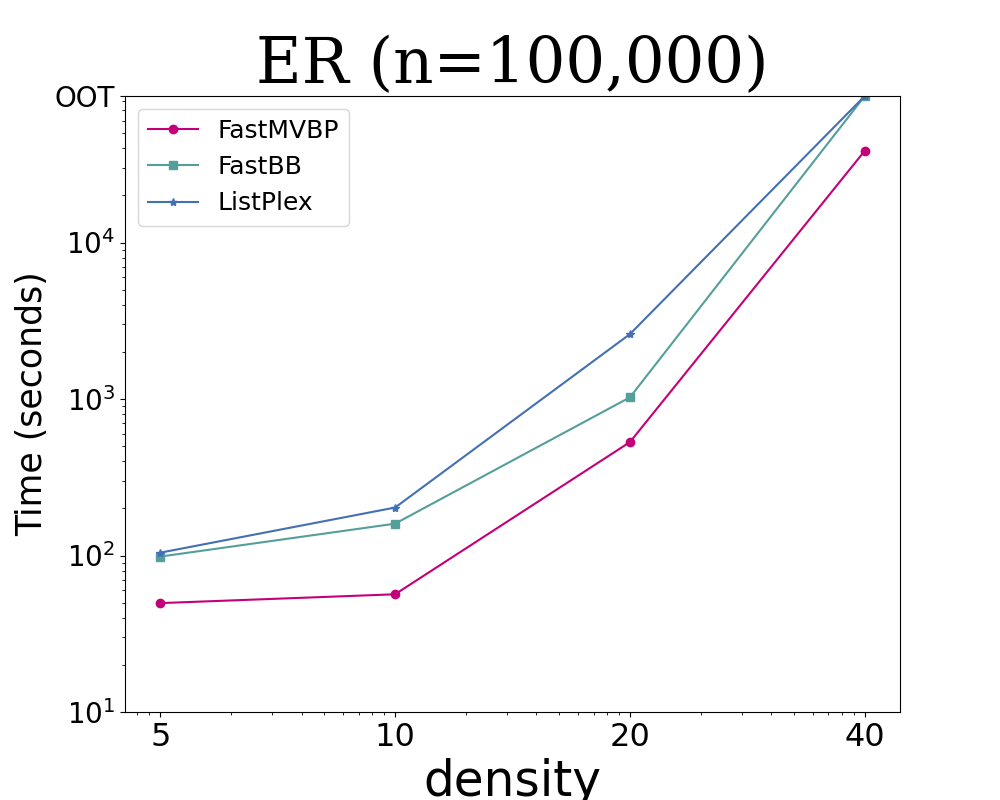}
	\end{minipage}
 \caption{The running time of listing top-$10$ maximal $s$-biplexes in different synthetic graphs}
 \label{fig:synthetic}
\end{figure}

\noindent
\subsection{Abalation study}
FastMVBP is equipped with three practical techniques: 2-hop decomposition, reduction with single-side bounds, and progressive search, as discussed in Section~\ref{section:techniques}.
In this study, we investigate the significance of each technique.

First, we observe that the single-side bounds play an indispensable role for solving the problems practically. Without using this technique in the reduction, the algorithm becomes OOT for almost all instances.
Therefore, we mainly investigate the decomposition and progressive search in the remaining.
We compare combinations of these two techniques by designing three variants:
{\em MVBP+R} denotes MVBP with single reduction technique, {\em MVBP+RD} employs both reduction and 2-hop decomposition, {\em MVBP+RP} utilises both reduction and progressive search, and FastMVBP employs all three techniques.

The results are given in Table~\ref{tab:abalation}. 
In general, to avoid out-of-time (OOT) errors, the 2-hop decomposition technique is more critical than the progressive search technique.
This might be attributed to that the 2-hop degeneracy $d_2$ is often small after applying reductions, which means the decomposition can prune a significant number of branches. 
We observe that algorithms utilizing the 2-hop decomposition technique(MVBP+RD and FastMVBP) often perform well on graphs with a small $d_2$.
The only exception is found in the Trackers dataset, where MVBP+RD performs worse than MVBP+RP. We also notice that a large $d_2$ value of 1850 is present in this dataset.
As for the progressive search technique, it tends to be effective in scenarios with small $\theta_L$ and $\theta_R$.
For instance, in the Youtube dataset, MVBP+RP outperforms other benchmarks when $\theta_L=\theta_R=3$, and is outperformed when $\theta_L$ and $\theta_R$ are large. 
This could be explained by that in these scenarios, there are a significant number of maximal $s$-biplexes that do not appear in the final solution set. 
By integrating the advantages of the three techniques, FastMVBP can efficiently solve instances that are OOT for ablation algorithms. 

\begin{table}[ht!]
    \centering
	\caption{ABALATION STUDY: The running time of listing top-10 maximal $s$-biplexes in real-world datasets}
	\label{tab:abalation}
    \resizebox{0.9\textwidth}{!}{
    \begin{tabular}{c|c|c|c|c|c|c|c}
	\toprule[1.5pt]
	\multirow{2}{*}{\tabincell{c}{Graph\\ $(|L|,|R|,|E|)$}} & \multirow{2}{*}{$s$}    & \multirow{2}{*}{$\theta_L,\theta_R$}  &  \multicolumn{4}{c|}{The running time (s)} & \multirow{2}{*}{$d_2$}
	\\
	\cline{4-7}
	& & &  MVBP+R & MVBP+RD & MVBP+RP &\textbf{FastMVBP} &\\
	\hline
 	\multirow{4}{*}{\shortstack{YouTube \\$(94238,30087,293360)$}} & \multirow{2}{*}{1} &  3,3 &OOT  &39256.1 & 236.063 & \textbf{10.644} &792\\
  & & 5,5 & OOT &306.333 &  914.491 &\textbf{4.286} &181\\
  \cline{2-8}
	 & 3 & 15,15 & 1168.08 & 257.585 & 880.095 & \textbf{196.715} &67\\
     \cline{2-8}
	& 5 & 20,20 & 14495.3 & \textbf{2641.18} & 14547.28& 2665.68&63 \\
	\cline{1-8}
  \multirow{3}{*}{\shortstack{Location\\ $(172091,53407,293697)$} } & \multirow{2}{*}{1} & 3,3 & OOT & 102.515 &20.007&\textbf{2.107}&695\\ 
    & & 5,5 & 1.209 &\textbf{0.020} &0.162 &0.050&32\\
 \cline{2-8}
	& 2 & 5,5 & 83.7543 & 87.108 & 11317.6 & \textbf{0.293}&32\\
 \cline{1-8}
	\multirow{4}{*}{\shortstack{CiteULike\\ $(153277,731769,2338554)$}} & \multirow{2}{*}{1} & 11,11 & OOT & 238.224 & 6764.79 & \textbf{132.364} &220\\
     & & 13,13 &54716 & 105.427& 12708.5&\textbf{75.204}&149\\
     \cline{2-8}
	& 3 & 20,20& 32271.8 & \textbf{17.23} & 7987.53 & 19.394 &149\\
      \cline{2-8}
	& 5 & 24,24 & 8732.66 & 529.221 & OOT & \textbf{267.457}&149\\
	\cline{1-8}
	\multirow{4}{*}{\shortstack{Twitter\\ $(175214,530418,1890661)$}} & \multirow{2}{*}{1} & 9,9& OOT & 105.737 & 11208.4 & \textbf{6.461}&171\\
      & & 11,11 &OOT &\textbf{16.191}& OOT&17.910& 96\\
      \cline{2-8}
	& 3 & 16,16& OOT & 31.321 & 7429.22 & \textbf{5.86} &53\\
      \cline{2-8}
	& 5 & 21,21 & 38705.6 & 4.527 & 3390.99 & \textbf{4.322} &55\\
	\cline{1-8}
	\multirow{4}{*}{\shortstack{IMDB\\ $(896302,303617,3782463)$}} & \multirow{2}{*}{1} & 3,3 & OOT & OOT & OOT & \textbf{125.608}&518\\
       & & 5,5 & OOT &703.135 &OOT &\textbf{4.926}&147\\
 \cline{2-8}
	& 3 & 13,13 & OOT & 74.96 & OOT & \textbf{112.613}&78\\
 \cline{2-8}
	& 5 & 25,25 & 2126.7 & \textbf{18.836} & 2355.18 & 19.090&59\\
	\cline{1-8}
 	\multirow{3}{*}{\shortstack{Team\\ $(901166,34461,1366466)$}} & \multirow{2}{*}{1}& 3,3 & OOT & 21203.7 & OOT & \textbf{765.875}&1114 \\
         & & 5,5 & OOT & 35.848 & OOT & \textbf{27.364}&234\\
	\cline{2-8}
	& 2 & 7,7 & OOT & \textbf{32.917} & OOT & 39.671 &239\\
 \cline{1-8}
 	\multirow{4}{*}{\shortstack{AmazonRatings\\ $(2146057,1230915,5743258)$}} & \multirow{2}{*}{1} & 7,7 & OOT & 617.842 & OOT & \textbf{9.567}&67\\
         & & 9,9 & OOT& 59.342& OOT&\textbf{16.923}&56\\
 \cline{2-8}
	& 3 & 11,11 & OOT & 205.805 & OOT & \textbf{22.221}&64\\
 \cline{2-8}
	& 5 & 17,17 & OOT & \textbf{955.473} & OOT & 962.669&73\\
    \cline{1-8}
	\multirow{3}{*}{\shortstack{Trackers\\$(27665730,12756244,140613762)$}} & \multirow{2}{*}{1} & 150,150 & OOT & OOT & 28.576 & \textbf{27.457}&1850\\
          & & 155,155 &OOT & 244.45& \textbf{28.347}&90.131&1850\\
 \cline{2-8}
	& 2 & 150,150 & OOT & OOT & 27.421 & \textbf{24.683}&1850\\
    \bottomrule[1.5pt]
    \end{tabular}
}
\end{table}

\noindent
\subsection{Case Study: Social Recommendation}
In this section, we study a real-world case where we apply the \emph{TBS} model in social recommendation. { The goal is to identify as many users as possible who share common interests and preferences, maximizing the coverage of relevant audiences. This enables more accurate recommendations and improves the effectiveness of targeted advertising.}
As a reference, we also test the top-$k$ $s$-biplex with most edges (TBSE) models \citep{yu2023maximum}. Recall that the TBSE asks for the $k$ maximal $s$-biplexs with most edges, where the maximality is still defined with vertex number.
Specifically, we use five social networks. 
The YouTube dataset contains users (w.l.o.g., denoted as \emph{left}) and their groups(denoted as \emph{right}).
The Twitter dataset contains users (left) and the tags (right).
The IMDB dataset contains movies (right) and actors (left) who participated in those movies.
The AmazonRatings dataset contains users (left) and products (right) rated by those users.
The Mooc dataset contains users (left) and course activities (right).
All these datasets are collected from SNAP\footnote{https://snap.stanford.edu/data} and KONECT\footnote{http://konect.cc}.
As shown in Table~\ref{tab:comprasion}, the TBS often identifies more users (left) who share common interests (right) than TBSE.
{ 
Specifically, in the Mooc dataset, TBS identifies additional 1114, 483 and 408 users than TBSE.}
{ To maximize the coverage of potential audiences for social recommendation, TBS is more preferred.}

\begin{table}[ht!]
    \centering
	\caption{Social Recommendation: identifying the users that share common interests ($k=1,s=1$)}
	\label{tab:comprasion}
    \resizebox{0.8\textwidth}{!}{
    \begin{tabular}{c|c|c|c|c|c|c|c}
	\toprule[2pt]
  \shortstack{Graph\\ $(|L|,|R|,|E|)$} & \shortstack{$\theta_L,\theta_R$} & 
 \shortstack{TBS\\(total)} & \shortstack{TBS\\(left)} & \shortstack{TBS\\(right)} &\shortstack{TBSE\\(total)}& \shortstack{TBSE\\(left)} &\shortstack{TBSE\\(right)}\\     
	\hline
 \multirow{3}{*}{\shortstack{Youtube\\ $(94238,30087,293360)$}} 
    & 8,8 & 36 &28& 8 & 34 & 9&25 \\ \cline{2-8}
	& 10,10 & 32 &10&22 & 31 & 11 &20\\ \cline{2-8}
	& 12,12 & 30&12&18& 30 & 14 &16\\  \cline{1-8}
\multirow{3}{*}{\shortstack{Twitter\\ 	 $(175214,530418,1890661)$}} 
    & 12,12& 46 & 34 & 12 & 44 & 29 &15\\ \cline{2-8}
    & 13,13 & 45 &32 &13  & 44 & 29 &15\\ 	\cline{2-8}
    & 14,14& 45 &31&14 & 44 & 29 &15\\       \cline{1-8}
\multirow{3}{*}{\shortstack{IMDB\\ $(896302,303617,3782463)$}} 
    & 3,3 & 192 &189 &3 & 192 &189 & 3  \\ \cline{2-8}
    & 4,4 & 129 &125 &4 & 95 &89 &6\\ \cline{2-8}
	& 5,5 & 111 & 106 &5 &95 &89 &6 \\ \cline{1-8}
\multirow{3}{*}{\shortstack{AmazonRatings\\ $(2146057,1230915,5743258)$}} 
    & 7,7 & 55 &48& 7 & 45 &33& 12\\ \cline{2-8}
	& 9,9 & 49 &40&9 & 45 & 33 &12\\ \cline{2-8}
	& 11,11 & 46 &35&11& 45 & 33 &12\\  \cline{1-8}
 \multirow{3}{*}{\shortstack{Mooc\footnotemark \\ $(7047,97,178443)$}} 
    & 3,3 & 5116 &5113& 3 & 4001 & 3997&4 \\ \cline{2-8}
	& 4,4 & 4484 &4480&4 & 4001 & 3997 &4\\ \cline{2-8}
	& 5,5 & 3850&3845&5& 3442 & 3437 &5\\  \cline{1-8}
    \bottomrule[2pt]
    \end{tabular}
}
\end{table} 
\footnotetext{To get TBSEs, we run the state-of-the-art algorithm in \citep{yu2023maximum} and report its best result before OOT.}

\section{Conclusion}
In this paper, we investigated the {\em TBS} problem aiming to find $k$ maximal $s$-biplexes with most vertices. There is a number of applications in graph data queries that ask for solving this problem efficiently.
We give a rigorous proof to show the NP-harness of the {\em TBS} problem.
An exact branching algorithm MVBP that solves the problem was presented afterward. 
We also investigated some effective techniques to improve the algorithm, including 2-hop decomposition, single-side bounds, and progressive search. Combining these techniques, we improved MVBP to FastMVBP which has a reduced time complexity both theoretically and empirically.
Experiments show that FastMVBP outperforms other benchmark algorithms, and in several instances, the speedup can be up to three orders of magnitude.
It is clear that many sub-instances in our algorithm can be computed independently due to the power of decomposition and branching techniques. We plan to explore efficient parallel computation in the future. Also, the proposed algorithms, as well as the analysis techniques, can be easily generalized to other subgraph query problems.

\section*{Acknowledgments}
This work was supported by National Natural Science Foundation of China under grant nos. 62372093 and 61972070, and Natural Science Foundation of Sichuan Province of China under grants 2023NSFSC1415 and 2023NSFSC0059. 


\bibliographystyle{elsarticle-num-names}
\bibliography{main}


\end{document}